\documentclass[journal]{IEEEtran}

\IEEEoverridecommandlockouts
\usepackage{bbding}
\usepackage{relsize}
\usepackage{moreverb}
\usepackage{mathrsfs}
\usepackage{amsmath}
\usepackage{amssymb}
\usepackage{stfloats}
\usepackage{graphicx}
\usepackage{makecell}
\usepackage{soul}
\soulregister{\cite}{7}
\soulregister{\ref}{7}
\usepackage{color,xcolor}
\usepackage{framed}
\usepackage{tikz}

\usepackage{cite}

\usepackage{algorithm}
\usepackage{algorithmicx}
\usepackage{algpseudocode}

\usepackage{multirow}       
\usepackage{booktabs}
\usepackage{multirow}
\usepackage{tabularx,ragged2e,booktabs,caption}

\ifCLASSOPTIONcompsoc
\usepackage[tight,normalsize,sf,SF]{subfigure}
\else
\usepackage[tight,footnotesize]{subfigure}
\usepackage{subfigure}

\usepackage{amsmath}
\usepackage{amsthm}

\newtheorem{proposition}{Proposition}
\usepackage{setspace}
\usepackage[hidelinks]{hyperref}
\usepackage{graphicx}
\usepackage{flushend}
\allowdisplaybreaks[3]

\soulregister{\emph}7
\soulregister{\cite}7
\UseRawInputEncoding

\begin{document}
	
	\title{\huge{Secure and Green RSMA-Assisted Heterogeneous ISAC}
}

{\Large{\author{\IEEEauthorblockN{Xudong Li, Rugui Yao, Theodoros A. Tsiftsis, and Alexandros-Apostolos A. Boulogeorgos}}}}
		%
	
	\maketitle
	
\begin{abstract}
This paper investigates sensing, communication, security, and energy efficiency of the heterogeneous integrated sensing and communication networks under challenging operational conditions. We focus on scenarios in which communication performance, security, and sensing accuracy are degraded by interference, eavesdropping, and imperfect channel state information.
To this end, we analyze communication and sensing signals within ISAC framework as well as the communication signals of a multicast network based on \emph{rate-splitting multiple access} (RSMA). Then, sensing signal-to-cluster-plus-noise ratio, communication rate, security rate, and \emph{security energy efficiency} (SEE) are evaluated.
To simultaneously enhance these system performances, we propose a targeted optimization framework aimed at maximizing SEE. This framework characterizes the sensing-security trade-off by jointly optimizing the transmit \emph{beamforming} (BF) vectors and the echo BF vector to construct green interference using the echo signal, as well as common and private streams generated by RSMA. Particularly, the joint design improves the security rate and reduces power consumption, thereby enabling a higher SEE. Given the non-convex nature of the optimization problem, we present an alternative approach that leverages Taylor series expansion, majorization-minimization, semi-definite programming, and successive convex approximation techniques. Specifically, we decompose the original non-convex and intractable optimization problem into three simplified sub-optimization problems, which are iteratively solved using an alternating optimization strategy.
Simulations provide comparisons with state-of-the-art schemes, highlighting the superior efficiency, robustness, and scalability of the proposed joint multi-BF optimization scheme based on RSMA and green interference in improving system performances.
\end{abstract}


\IEEEpeerreviewmaketitle

\section{Introduction}

\IEEEPARstart{W}{ith} the rapid promotion of \emph{fifth generation} (5G), \emph{beyond} (B5G) and even \emph{sixth generation} (6G) networks, spectrum scarcity becomes the key limited factor of next-generation bandwidth-hungry applications \cite{ISAC_PLS_6_re}. Given that the current frequency band utilization is in the range of 15-20\% \cite{CR_data}, it is necessary to explore emerging technologies that refine spectrum utilization and mitigate radio resource competitions, such as \emph{integrated sensing and communication} (ISAC) \cite{ISAC_PLS_4,ISAC_PLS_5,ISAC_PLS_52}. 

The joint or independent optimization of sensing and communication in the ISAC has emerged as a key research focus. \cite{ISAC_PLS_15_r,ISAC_PLS_16_r,ISAC_PLS_7}. The authors of \cite{ISAC_PLS_30} studied the single-static sensing performance of the multi-target massive \emph{massive input massive output} (MIMO)-ISAC systems, and minimized the sum of \emph{Cramer-Rao lower bounds} (CRLBs) of target arrival directions under a communication rate constraint. The authors stated that their scheme achieved near-optimal performance with less complexity through a single optimization of the sensing signal.
In \cite{ISAC_PLS_19}, the authors investigated a mobile-antenna-assisted ISAC system, aiming to improve the communication rate and echo signal \emph{signal-to-cluster-plus-noise ratio} (SCNR) by jointly optimizing the antenna coefficient and antenna position. An ISAC network assisted by a fluid antenna was designed in \cite{ISAC_PLS_24}. With the position and waveform of the fluid antenna jointly optimized, the sum communication rate of all users was improved. The authors of \cite{ISAC_PLS_20} presented an ISAC network in which \emph{unmanned aerial vehicle} (UAV) acted as a \emph{base station} (BS). To achieve a balance between sensing and communication, the position and transmit power of the UAV were optimized, and then the communication rate and CRLB of target sensing accuracy were maximized. 
Li \emph{et al.} analyzed an ISAC network in which a mobile UAV served as the sensing target. Under the communication \emph{quality of service} (QoS) and transmit power of the BS constraints, transmit signal BF vector and target allocation were optimized to improve the echo signal \emph{signal-to-interference-plus-noise ratio} (SINR) \cite{ISAC_PLS_23}. In \cite{ISAC_PLS_21}, the author studied ISAC composed of multiple BSs and multiple users acting as targets and communication users. 

 Considering that the overall composition in ISAC is relatively complex and there is a non-negligible interference between the target sensing and the multicast communications, the scientific community turned its attention towards a variety of multiple access schemes to attain interference mitigation in ISAC \cite{MA_ISAC,Inter_ISAC}. The authors of \cite{ISAC_PLS_42} focused on a \emph{space division multiple access} (SDMA) scheme that used linear precoding to distinguish users in a spatial domain, relying entirely on treating any remaining multi-user interference as noise. The communication performance gain of the ISAC based on SDMA was evaluated in \cite{ISAC_PLS_43}. In contrast to SDMA, \emph{non-orthogonal multiple access} (NOMA) works with superimposing coding at the transmitter and \emph{successive interference cancellation} (SIC) coding at the receiver. NOMA superimposes users in the power domain, and forces users with better channel conditions to perform complete decoding through user grouping and sorting \cite{ISAC_PLS_42} to eliminate interference caused by other users. Different from NOMA, \emph{orthogonal multiple access} (OMA) assigns one resource to one user, resulting in lower resource utilization. 
The authors of \cite{ISAC_PLS_44} assessed the OMA-empowered and NOMA-empowered performance gains of communication and sensing of the semi-ISAC network based on the traversal rate and the traversal estimation of information rate. An uplink transmission scheme was articulated in \cite{ISAC_PLS_45} for NOMA-ISAC system to mitigate mutual interference between sensing and communication signals, and enhance communication convergence rate, reliability, and sensing accuracy. In \cite{ISAC_PLS_46_r}, the authors documented a joint optimization scheme of transmit signal BF, NOMA transmission time, and target sensing scheduling to maximize the sensing efficiency of ISAC systems, while ensuring a high communication QoS. 
A joint precoding optimization problem based on NOMA was solved in \cite{ISAC_PLS_47}, which maximized the security rate of multi-user through \emph{artificial noise} (AN), and achieved secure transmission, while satisfying the sensing performance constraint. 

Although studies on SDMA and NOMA to improve the ISAC performance is gradually deepening, there are some extremes to conventional multi-access architectures, such as SDMA and NOMA. Specifically, SDMA treats interference entirely as noise, seriously reducing the reliability. Instead, NOMA decodes interference one by one, implying that the effectiveness is hard to guarantee. Above shortcomings and deficiencies urge us to find a new scheme like \emph{rate-splitting multiple access} (RSMA)  \cite{ISAC_PLS_14}, adopting rate splitting based on linear precoding and SIC. As a consequence, RSMA decodes some interference and treats the remaining as noise, fully absorbing advantages of both SDMA and NOMA, and achieves high reliability and high effectiveness. Under the constraints of data rate and transmit power budget, RSMA structure and parameters were designed in \cite{ISAC_PLS_48} to minimize the CRLB of sensing response matrix at radar receivers. The authors of \cite{ISAC_PLS_49} presented an indicative example of an RSMA-ISAC waveform design that jointly optimized the minimum fairness rate among communication users and the CRLB of target detection under power constraints.

Coexistence of sensing and communication broadens the prospects for next-generation communication systems, while increasing the energy consumption. This necessitates the development of green ISAC systems simultaneously. In this direction, the authors of \cite{ISAC_PLS_29} introduced a power consumption minimization policy for the near-field ISAC system. In particular, the transmit signal BF vector was optimized to minimize network power consumption under the constraints of communication SINR, sensing target transmit beam pattern gain, and interference power. For \emph{intelligent reflecting surface} (IRS)-ISAC systems, the authors of \cite{ISAC_PLS_25_re,ISAC_PLS_36_re} maximized the \emph{energy efficiency} (EE) by jointly optimizing the transmit signal BF vector, the IRS reflection coefficient matrix, and the IRS deployment location. Energy-saving BF design of ISAC systems that aimed to maximize the EE by appropriately designing transmission waveforms in multi-user communication and target estimation scenarios was documented in \cite{ISAC_PLS_37} and \cite{ISAC_PLS_38}. 
In a multi-BS ISAC network, the energy consumption was reduced due to optimum task allocation, beam scheduling and transmit power control \cite{ISAC_PLS_40}. 

Additionally, due to the inherent open nature of downlink data transmission and broadcast mechanism, as well as the resource sharing between perception and communication of the ISAC network, it is vulnerable to security threats like eavesdropping and intercepting \cite{Intro_PLS_ISAC}. Consequently, it is of great significance and urgency to carry out researches on \emph{physical layer security} (PLS) in ISAC networks \cite{ISAC_PLS_50}. In an IRS-ISAC network, the authors of \cite{ISAC_PLS_27} maximized the minimum communication rate by optimizing the transmit BF vector, the receive BF vector, and the IRS reflection coefficient matrix under the constraints of echo signal power and security rate. For the same system model, BF was designed to maximize the minimum weighted beam pattern gain under security rate and transmit power constraints \cite{ISAC_PLS_33}. 
The authors of \cite{ISAC_PLS_28} focused on an ISAC network in which UAVs served as BSs to provided downlink data transmission for multiple users, sense and interfere with the \emph{eavesdropper} (Eve) to maximize security sum rate. 
The authors of \cite{ISAC_PLS_32} used neural networks to optimize the transmit signal precoders to minimize the maximum SINR of the Eve. In a UAV-IRS-ISAC system, the DRL framework was employed to optimize the transmit signal BF vector and the coefficient matrix of the IRS loaded by a UAV to maximize the security sum rate \cite{ISAC_PLS_34}. In \cite{ISAC_PLS_35}, NOMA and AN were adopted to jointly optimize the radar correlation and transmit signal BF vector to maximize echo signal power. 

To sum up, \emph{state-of-the-art} (SOTA) ISAC works \cite{ISAC_PLS_23,ISAC_PLS_30,ISAC_PLS_33,ISAC_PLS_35,ISAC_PLS_48} improved sensing of the ISAC system by optimizing echo signal power, SINR, or target detection CRLB. In \cite{ISAC_PLS_25_re,ISAC_PLS_27,ISAC_PLS_28,ISAC_PLS_32,ISAC_PLS_34,ISAC_PLS_47}, the security or communication QoS was enhanced by optimizing the communication sum rate, security coordination rate, security rate, or eavesdropping SINR. Optimization frameworks for jointly enhancing sensing and communication were articulated in \cite{ISAC_PLS_19,ISAC_PLS_20,ISAC_PLS_21,ISAC_PLS_44,ISAC_PLS_45,ISAC_PLS_49}. In \cite{ISAC_PLS_29,ISAC_PLS_36_re,ISAC_PLS_37,ISAC_PLS_38,ISAC_PLS_40}, the objective was to refine the EE. From the above, it gets obvious that joint optimization of multiple performances like sensing, communication, security, and EE of the ISAC has not been extensively investigated.

In PLS-ISAC, the dynamic change of channel state and the non-cooperative characteristics of unauthenticated Eve make the acquisition of the perfect CSI extremely difficult, while \cite{ISAC_PLS_28,ISAC_PLS_33,ISAC_PLS_34,ISAC_PLS_47} assume that Eve's perfect CSI can be ascertained. Undoubtedly, the assumption of perfect CSI provides tractability allow the derivation of performance bound with the cost accuracy and applicability as it. 
The aforementioned contributions on the PLS-ISAC \cite{ISAC_PLS_27,ISAC_PLS_28,ISAC_PLS_32,ISAC_PLS_33,ISAC_PLS_34,ISAC_PLS_35} take insufficient advantage of the inherent interference in networks. From the perspective of security, AN improves the security rate. Meanwhile, it results in notable power overhead and computational complexity increase.

Based on the aforementioned discussions, this paper focuses on the heterogeneous ISAC (i.e., ISAC network and the multicast communication network coexist), and provides a detailed and comparative analysis on several representative works summarized in Table \ref{work_com}.
\begin{table*}[htbp]
	\centering
	\caption{\label{work_com}Comparison on several representative works.}
	\scalebox{0.88}{
	\begin{tabular}{c|c|c|c|c|c}
		\hline\hline 
		\textbf{Item} & \textbf{Scheme} & \textbf{Network architecture} & \textbf{Security-aware} & \textbf{\makecell{Targeted optimization}}& \textbf{Feature}\\ 	\hline\hline
		This paper	& RSMA  &HISAC 
		with two BSs
		&\Checkmark & SEE	& \makecell{Good and balanced overall performance \\ in SEE, robustness, and scalability.}  \\ \hline
		[17]	& SDMA  &Single-cell BS &\XSolid & Communication rate	& \makecell{Efficient spectrum utilization.}  \\ \hline
		[19]	& OMA  &Single-cell ISAC &\XSolid & Communication rate	& \makecell{Inefficient spectrum utilization.}  \\ \hline
		[21] & NOMA & Single-cell ISAC
		&\XSolid & \makecell{Sensing efficiency} & \makecell{Medium performance, \\
			but limited in user scalability.} \\ 	\hline
		[24]	& RSMA & Single-cell ISAC
		&\XSolid & \makecell{Sensing CRB}  & \makecell{Good communication, but the \\security mechanism is not perfect.}  \\	\hline
		[35]	& N/A & \makecell{Single-cell ISAC
		with RIS}
		&\Checkmark & \makecell{Minimum weighted \\
		beampattern gain}
		 & \makecell{Excellent sensing, but hardware \\requirements and complexity are highest.} \\ 	\hline\hline 
	\end{tabular} }
\end{table*}
The proposed heterogeneous ISAC architecture with interfering systems reflects real-world deployment scenarios where multiple communication and sensing systems inevitably coexist and interact. On the one hand, with the scarcity of spectrum resources, multiple systems must operate in overlapping frequency bands, creating inherent interference scenarios.  On the other hand, current 5G-advanced and next generation 6G networks will deploy ISAC capabilities alongside existing communication infrastructure, creating the exact heterogeneous scenario we study and investigate.  Separate dedicated networks for sensing and communication are economically prohibitive, making integrated heterogeneous architectures the practical solution. Different from NOMA, SDMA, OMA, RSMA provides fundamental advantages for the heterogeneous ISAC scenario.  RSMA operates based on the principle of rate-splitting and partial interference decoding, which has been shown to achieve the capacity region of interference channels in certain settings.  Owing to the structure of RSMA, part of the interference is decoded (common stream) while the remainder is treated as noise.  It is less sensitive to imperfect CSI compared to NOMA, SDMA, and OMA, and adapts naturally to heterogeneous QoS requirements. For the considered heterogeneous ISAC, RSMA provides unique benefits for physical layer security. For green interference construction, the common stream in RSMA naturally creates structured interference that can be optimized to degrade eavesdropper performance without requiring additional power allocation, maintains legitimate user quality through intelligent rate-splitting, and enhances SEE by leveraging existing communication signals rather than dedicated jamming.

We construct green interference and articulate a joint BF optimization scheme based on RSMA in order to jointly enhance sensing SCNR, security rate, and SEE defined as the security data transmission rate achieved per unit of power.
In more detail, our contributions can be summarized as follows:
\begin{enumerate} 
	\item We introduce the \emph{heterogeneous ISAC} (HISAC) architecture supporting both ISAC and multicast communications, which increases the utilization of limited spectrum resources and improves performances despite adverse conditions such as interference, eavesdropping, and imperfect CSI of target, Bob, and Eve are confronted. Building upon this, we quantify sensing, communication, security, and power consumption through echo signal SCNR, security rate, and SEE. 
	
	\item We propose an optimization framework designed to maximize the SEE, and study the trade-off between sensing SCNR and security rate by optimizing both the transmit BF vectors and the echo BF vector. Our approach leverages RSMA and ISAC signals to generate green interference from multicast communication signals. This green interference is strategically designed to minimize interference for Bob while increasing interference for Eve; thereby, enhancing security rate. Unlike AN, green interference is an inherent component of the system's energy, eliminating the need for additional power consumption and computational complexity. Hence, this approach contributes to the overall improvement of SEE.
	
	\item We recognize that the original optimization objective and its associated constraints are inherently non-convex. To address this challenge, we propose an optimization algorithm that leverages Taylor series expansion, \emph{majorization-minimization} (MM), \emph{semi-definite programming} (SDP), and \emph{successive convex approximation} (SCA). This approach systematically transforms the originally non-convex and intractable problem into three more manageable sub-optimization problems. By employing an alternating iteration strategy, our method effectively converges toward a solution for the original optimization problem.
	
	
\end{enumerate}
\begin{figure}[t]
	
	\centering
	\includegraphics[width=8.625cm,height=4.6cm]{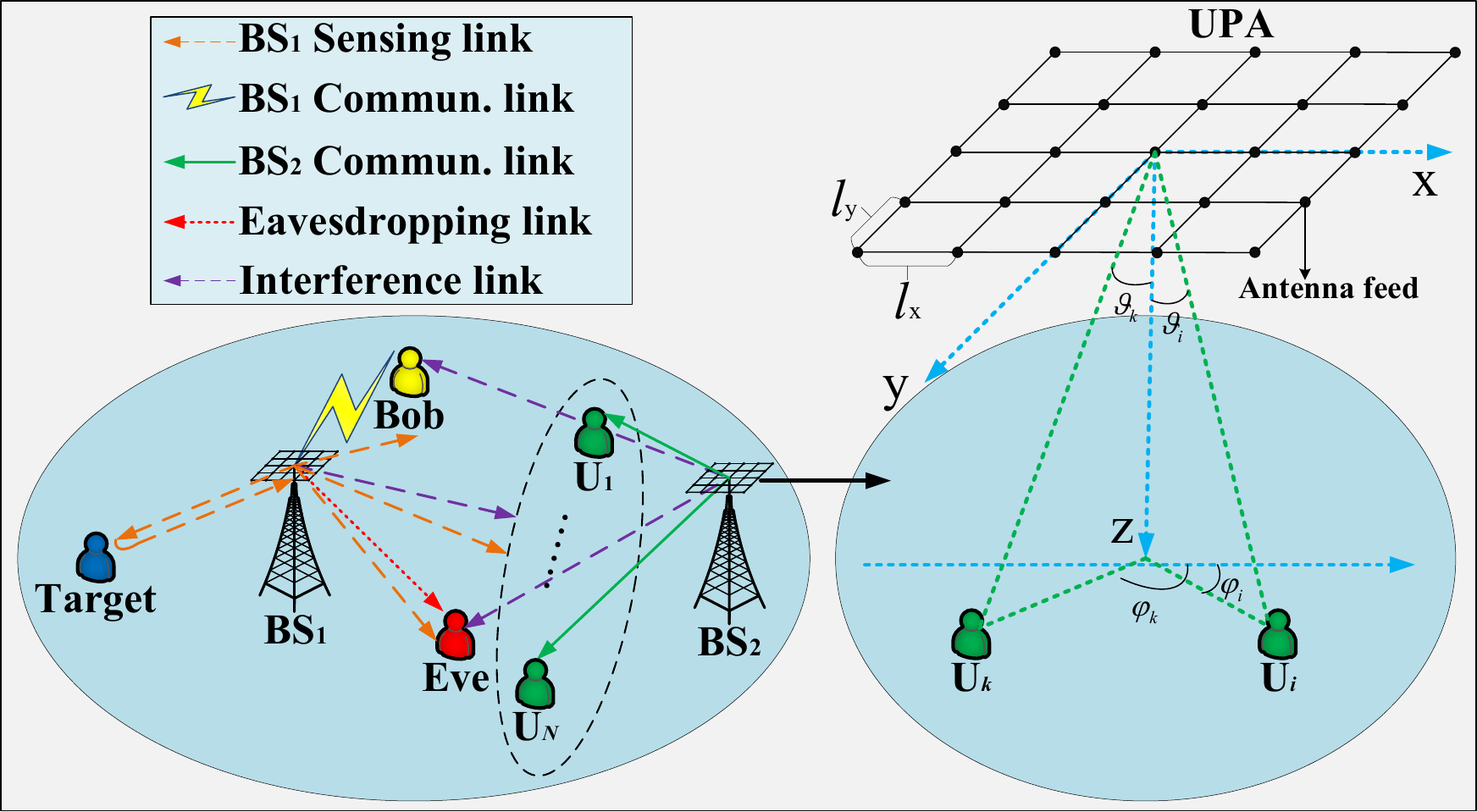}
	\caption{HISAC architecture with interference, eavesdropping, as well as imperfect CSI.}\label{fig_system_model} 
\end{figure}

The remainder of the paper is organized as follows: The channel model, heterogeneous ISAC architecture, and signal models accompanied by performance metrics are provided in Section II. Section III formulates the optimization problem, and present the corresponding solutions. Numerical results and simulations are provided in Section IV. Section V concludes this paper by summarizing its main message and key remarks.

\textbf{Notations:} 
Matrices and vectors are denoted as uppercase boldface and lowercase boldface, respectively. ${\left(  \cdot  \right)^{\rm{H}}}$, ${\mathop{\rm Tr}\nolimits} \left(  \cdot  \right)$, ${\mathop{\rm rank}\nolimits} \left(  \cdot  \right)$, and ${\left\|  \cdot  \right\|_2}$ are the Hermitian transpose operation, trace operation, rank operation, and 2-norm operation. ${\mathbb{C}^{x \times y}}$ stands for the 2-dimension complex space. ${{\bf{I}}_{Z \times Z}}$ represents the $Z \times Z$ identity matrix, ${{\bf{I}}_{{Z \times 1}}}$ represents the $\left( {{Z} \times 1} \right)$-dimension unit vector, $ \otimes $ denotes the Kroneker product of two matrices, $\mathcal{O}$ is computational complexity, and ${\left[ x \right]^ + } = \max \left\{ {0,x} \right\}$.

\section{System, Channel, and Signal Models}
To facilitate the formulation and solution of the optimization problem, we describe the architecture, Rician-shadowed fading channel, legitimate communication signal and unauthenticated eavesdropping signal in the ISAC, the RSMA-based multicast communication signal, and the sensing signal in the ISAC.

\subsection{System Model}
As shown in Fig. \ref{fig_system_model}, in the heterogeneous ISAC architecture, there are two sub-networks, namely ISAC as the primary network and multicast communications as the secondary network. 
In the primary network, $\rm BS_1$ is equipped with $M_1$ antennas provides the downlink data transmission service to Bob while sensing the target's location status information. Meanwhile, Eve, an unauthorized external eavesdropper equipped with a single antenna, attempts to wiretap confidential information that $\rm BS_1$ transmits to Bob.
In the secondary network, $\rm BS_2$ equipped with $M_2$ antennas employs the RSMA scheme to provide highly reliable and low-interference downlink data transmission for $N$ users. Clearly, in the heterogeneous ISAC, the signal transmitted by $\rm BS_1$ interferes with $N$ communication users' received signals, while the signal transmitted by $\rm BS_2$ affects both Bob's and Eve's signal reception.

\subsection{Channel Model}
In the primary network, a portion of the \emph{line of sight} (LoS) component is refracted and scattered by buildings and trees, forming the \emph{non LoS} (NLoS) component that coexists with the LoS component.  Therefore, the downlink channel can be modeled as the superposition of a predominant LoS component and a sparse set of single-bounce
NLoS components. This situation can be accurately characterized by the Rician-shadowed fading channel model in \cite{ISAC_PLS_53}, which is taken into account to capture the statistical characteristics of communication and sensing. Specifically, we have
\begin{equation}\label{channel_model}
\begin{array}{l}
	{\bf{h}} = \sqrt {\rho \left( {{\vartheta _0},{\varphi _0}} \right)} {\alpha _{\rm{0}}}{{\bf{u}}_{\rm{h}}}\left( {{\vartheta _0},{\varphi _0}} \right) \otimes {{\bf{u}}_{\rm{e}}}\left( {{\vartheta _0},{\varphi _0}} \right) + \\ \displaystyle
	\frac{1}{{\sqrt T }}\sum\nolimits_{t = 1}^T {\sqrt {\rho \left( {{\vartheta _t},{\varphi _t}} \right)} {\alpha _t}{{\bf{u}}_{\rm{h}}}\left( {{\vartheta _t},{\varphi _t}} \right) \otimes {{\bf{u}}_{\rm{e}}}\left( {{\vartheta _t},{\varphi _t}} \right)} ,
\end{array}
\end{equation}
where $\rho\left( {\vartheta ,\varphi } \right)$ is the antenna directivity pattern, ${{\bf{u}}_{\rm{h}}}\left( {\vartheta ,\varphi } \right)$ and ${{\bf{u}}_{\rm{e}}}\left( {\vartheta ,\varphi } \right)$ are horizon and elevation steering vectors, respectively, $\alpha_0$ and $\alpha_t$ are path-loss parameters of the LoS and the $t$-th NLoS components, $\varphi $ and $\vartheta $ are the azimuth and elevation angles of departure, respectively, $T$ is the number of NLoS components, $l_{\rm x}$ and $l_{\rm y}$ are distances between two adjacent antenna feeds of horizon direction and elevation direction, and ${M_1} = {M_{{\rm{1,1}}}} \times {M_{{\rm{1,2}}}}$.
In what follows, we point out that  unless otherwise specified, all channels considered follow Rician-shadowed distributions by default. While the Nakagami-$m$ fading model offers a versatile alternative, the Rician-shadowed fading is chosen for its physical accuracy in characterizing the dominant but potentially shadowed LoS paths critical for sensing and primary communication performance.

\subsection{Communication Signal Model}
For $\rm BS_1$, let ${x_{{\rm{bo}}}}$ and ${x_{{\rm{ta}}}}$ represent the communication signal transmitted to Bob and the sensing signal transmitted to target, respectively. ${{\bf{w}}_{{\rm{bo}}}} \in {\mathbb{C}^{{M_1} \times 1}}$ and ${{\bf{w}}_{{\rm{ta}}}} \in {\mathbb{C}^{{M_1} \times 1}}$ are the BF vectors of ${x_{{\rm{bo}}}}$ and ${x_{{\rm{ta}}}}$, respectively. Before transmitting ${x_{{\rm{bo}}}}$ and ${x_{{\rm{ta}}}}$, $\rm BS_1$ superimposes ${x_{{\rm{bo}}}}$ and ${x_{{\rm{ta}}}}$ together. Thus, the final signal transmitted by $\rm BS_1$ can be expressed as
\begin{equation}\label{B1_signal}
{{\bf{x}}_1} = {{\bf{w}}_{{\rm{bo}}}}{x_{{\rm{bo}}}} + {{\bf{w}}_{{\rm{ta}}}}{x_{{\rm{ta}}}}.
\end{equation}
Additionally, given that $\rm BS_2$ uses the RSMA scheme to implement multicast communications, all signals transmitted  by $\rm BS_2$ to $N$ communication users can be divided into two parts, i.e.,  common stream and the private stream, namely $\left\{ {z_{\rm 1}^{\rm{c}},z_{\rm 1}^{\rm{p}}} \right\}$, $\left\{ {z_{\rm 2}^{\rm{c}},z_{\rm 2}^{\rm{p}}} \right\}$, $ \cdots  \cdots $, $\left\{ {z_{{\rm }N}^{\rm{c}},z_{{\rm },N}^{\rm{p}}} \right\}$, respectively. 
Then, $N$ common streams, $\left\{ {z_{\rm 1}^{\rm{c}},z_{\rm 2}^{\rm{c}}, \cdots ,z_{\rm \it N}^{\rm{c}}} \right\}$, are extracted,  combined, and encoded as the common information, $s_{\rm c}$ via a codebook shared by $N$ communication users.
Private streams, $\left\{ {z_{\rm 1}^{\rm{p}},z_{\rm 2}^{\rm{p}}, \cdots ,z_{\rm \it N}^{\rm{p}}} \right\}$, are encoded separately as private information $\left\{ {{s_{ 1}},{s_{ 2}}, \cdots ,{s_{ N}}} \right\}$, respectively. It is assumed that the signals $x_{\rm bo}$, $x_{\rm ta}$, $s_{\rm c}$, and $s_n$ are uncorrelated with zero mean and unit variance.
Then, we assume that ${{\bf{o}}_{\rm{c}}}\in {\mathbb{C}^{M_2 \times 1}}$ and ${{\bf{o}}_{{\it n}}}\in {\mathbb{C}^{M_2 \times 1}}$ are beamformers of the common information $s_{\rm i,c}$ and the private information of the $n$-th communication user ${s_{ n}}$, $n \in \left\{ {1,2, \cdots ,N} \right\}$, respectively. 
The signal transmitted by ${\rm BS_ 2}$ can be
\begin{equation}\label{BS2_signal}
{{\bf{x}}_2} = {{\bf{o}}_{\rm{c}}}{s_{\rm{c}}} + \sum\nolimits_{n = 1}^N {{{\bf{o}}_n}{s_n}}.
\end{equation}
Let ${{\bf{h}}_{{\rm{bo}}}} \in {\mathbb{C}^{{M_1} \times 1}}$, ${{\bf{h}}_{{\rm{ta}}}} \in {\mathbb{C}^{{M_1} \times 1}}$ ${{\bf{h}}_{{\rm{e}}}} \in {\mathbb{C}^{{M_1} \times 1}}$, and ${{\bf{h}}_{{n}}} \in {\mathbb{C}^{{M_1} \times 1}}$ represent the channels from $\rm BS_1 $ to Bob, Eve, and the $n$-th user. Likewise, ${{\bf{g}}_{{\rm{bo}}}} \in {\mathbb{C}^{{M_2} \times 1}}$, ${{\bf{g}}_{{\rm{e}}}} \in {\mathbb{C}^{{M_2} \times 1}}$, and ${{\bf{g}}_{{n}}} \in {\mathbb{C}^{{M_2} \times 1}}$ represent the channels from $\rm BS_2 $ to Bob, Eve, and the $n$-th user. $n_{\rm bo}$, $n_{\rm e}$, and $n_{ n}$ are independent and identically distributed complex random Gaussian
noises with mean being 0 and variance being $\sigma^2$ at Bob, Eve, and $n$-th user.
Received signals at Bob, Eve, and $n$-th user are given by
\begin{equation}\label{Bob_signal}
\begin{array}{*{20}{l}} \displaystyle
	{{y_{{\rm{bo}}}} = {\bf{h}}_{{\rm{bo}}}^{\rm{H}}{{\bf{w}}_{{\rm{bo}}}}{x_{{\rm{bo}}}} + {n_{{\rm{bo}}}}}\\ \quad\;\;\, \displaystyle
	{ + {\bf{h}}_{{\rm{bo}}}^{\rm{H}}{{\bf{w}}_{{\rm{ta}}}}{x_{{\rm{ta}}}} + {\bf{g}}_{{\rm{bo}}}^{\rm{H}}{{\bf{o}}_{\rm{c}}}{s_{\rm{c}}} + \sum\nolimits_{n = 1}^N {{\bf{g}}_{{\rm{bo}}}^{\rm{H}}{{\bf{o}}_n}{s_n}} ,}
\end{array}
\end{equation}
\begin{equation}\label{Eve_signal}
\begin{array}{*{20}{l}} \displaystyle
	{{y_{\rm{e}}}{\rm{ = }}{\bf{h}}_{\rm{e}}^{\rm{H}}{{\bf{w}}_{{\rm{bo}}}}{x_{{\rm{bo}}}} + {n_{\rm{e}}}}\\ \displaystyle \quad
	{ + {\bf{h}}_{\rm{e}}^{\rm{H}}{{\bf{w}}_{{\rm{ta}}}}{x_{{\rm{ta}}}} + {\bf{g}}_{\rm{e}}^{\rm{H}}{{\bf{o}}_{\rm{c}}}{s_{\rm{c}}} + \sum\nolimits_{n = 1}^N {{\bf{g}}_{\rm{e}}^{\rm{H}}{{\bf{o}}_n}{s_n}} ,}
\end{array}
\end{equation}
\begin{equation}\label{n-th-user_signal}
\begin{array}{*{20}{l}} \displaystyle
	{{y_n} = {\bf{g}}_n^{\rm{H}}{{\bf{o}}_{\rm{c}}}{s_{\rm{c}}} + {\bf{g}}_n^{\rm{H}}{{\bf{o}}_n}{s_n} + {n_n}}\\ \displaystyle  \quad\;
+	{\sum\nolimits_{k \ne n} {{\bf{g}}_n^{\rm{H}}{{\bf{o}}_k}{s_k}}  + {\bf{h}}_n^{\rm{H}}{{\bf{w}}_{{\rm{bo}}}}{x_{{\rm{bo}}}} + {\bf{h}}_n^{\rm{H}}{{\bf{w}}_{{\rm{ta}}}}{x_{{\rm{ta}}}},}
\end{array}
\end{equation}
respectively. 
In practice, Eve’s channel is inherently difficult to estimate accurately. Assuming perfect or average-case CSI for Eve is impractical and risky, as it could lead to an overestimation of the achieved security rate. The worst-case approach considering the error bound ensures the designed beamformers guarantee a minimum security rate even under the most unfavorable channel estimation error within the bounded set. Likewise, while Bob may have more precise CSI through estimation and feedback, employing a robust design for these links ensures QoS constraints are satisfied despite estimation errors, enhancing system reliability. Besides, the bounded error model coupled with the triangle inequality can transform an intractable probabilistic problem into a deterministic one efficiently. Hence, we define imperfect CSI for $\rm BS_1$-$i$ and $\rm BS_2$-$i$ links. As such, we have ${{\bf{h}}_{i}} = {{\bf{h}}_{i,{\rm{es}}}} + {{\bf{h}}_{i,{\rm{er}}}}$ and ${{\bf{g}}_{i}} = {{\bf{g}}_{i,{\rm{es}}}} + {{\bf{g}}_{i,{\rm{er}}}}$, respectively, where ${{\bf{h}}_{i,{\rm{es}}}}$ and ${{\bf{g}}_{i,{\rm{es}}}}$ are the estimated CSI of the corresponding links, ${{\bf{h}}_{i,{\rm{er}}}}$ and ${{\bf{g}}_{i,{\rm{er}}}}$ represent the differences between the real and estimated CSI, of which the 2-norms satisfy $0 \le {\left\| {{{\bf{h}}_{i,{\rm{er}}}}} \right\|_2} \le {{{e}}_{\rm{h}}}$ as well as $0 \le {\left\| {{{\bf{g}}_{i,{\rm{er}}}}} \right\|_2} \le {{{e}}_{\rm{g}}}$, $i \in \left\{ {{\rm{bo,e}}} \right\}$, where ${{e}_{\rm{h}}}$ and ${{{e}}_{\rm{g}}}$ are the upper bounds of 2-norms. In what follows, given that the derivation of Bob's imperfect CSI is similar to that of Eve's imperfect CSI and to concisely demonstrate the formula derivation, we take the CSI uncertainty derivation of the Eve as an example.
To quantify the CSI uncertainty of $\rm BS_1$-Eve link, there exist
\begin{equation}\label{BS1_Eve_bo_uncertainty}
{\left\| {{\bf{h}}_{\rm{e}}^{\rm{H}}{{\bf{w}}_{{\rm{bo}}}}} \right\|_2^2 = {\bf{w}}_{{\rm{bo}}}^{\rm{H}}\left( {{{\bf{h}}_{{\rm{es}}}}{\bf{h}}_{{\rm{es}}}^{\rm{H}} + {{\bf{\Delta }}_{\rm{h}}}} \right){{\bf{w}}_{{\rm{bo}}}},}
\end{equation}
\begin{equation}\label{BS1_Eve_ta_uncertainty}
	\begin{array}{l}
		\left\| {{\bf{h}}_{\rm{e}}^{\rm{H}}{{\bf{w}}_{{\rm{ta}}}}} \right\|_2^2 
		= {\bf{w}}_{{\rm{ta}}}^{\rm{H}}\left( {{{\bf{h}}_{{\rm{es}}}}{\bf{h}}_{{\rm{es}}}^{\rm{H}} + {{\bf{\Delta }}_{\rm{h}}}} \right){{\bf{w}}_{{\rm{ta}}}},
	\end{array}
\end{equation}
where ${{\bf{\Delta }}_{\rm{h}}} = {{\bf{h}}_{{\rm{es}}}}{\bf{h}}_{{\rm{er}}}^{\rm{H}} + {{\bf{h}}_{{\rm{er}}}}{\bf{h}}_{{\rm{es}}}^{\rm{H}} + {{\bf{h}}_{{\rm{er}}}}{\bf{h}}_{{\rm{er}}}^{\rm{H}}$ is the CSI error matrix of $\rm BS_1$-Eve link, and satisfies the compatibility and triangle inequality constraint. 
The bounds for ${{\bf{\Delta }}_{\rm{h}}}$ and ${{\bf{\Delta }}_{\rm{g}}}$ are derived by applying the triangle inequality to the norm of the sum of rank-one matrices, leading to the corresponding upper bounds $e_{\rm h,UB}$ and $ e_{\rm g,UB}$. More details are provided in Appendix A.
The aforementioned error bounds used in ${{{\left\| {{{\bf{\Delta }}_{\rm{h}}}} \right\|}_2} \le  {{{e}}_{{\rm{h}},{\rm{UB}}}}}$ and ${\left\| {{{\bf{\Delta }}_{\rm{g}}}} \right\|_2} \le  {{{e}}_{{\rm{g,UB}}}}$ in the derivations of ${{\bf{\Delta }}_{\rm{h}}}$ and ${{\bf{\Delta }}_{\rm{g}}}$ are not arbitrarily chosen but are direct mathematical consequences of applying the triangle inequality to the original error model. They represent the maximum possible effect of the bounded errors $ e_{\rm h}$ and $ e_{\rm g}$ on the quadratic terms 	$\left\| {{\bf{h}}_{\rm{e}}^{\rm{H}}{{\bf{w}}_{{\rm{ta}}}}} \right\|_2^2$, $\left\| {{\bf{h}}_{\rm{e}}^{\rm{H}}{{\bf{w}}_{{\rm{bo}}}}} \right\|_2^2$, $\left\| {{\bf{g}}_{\rm{e}}^{\rm{H}}{{\bf{o}}_{{\rm{c}}}}} \right\|_2^2$, and $\left\| {{\bf{g}}_{\rm{e}}^{\rm{H}}{{\bf{o}}_{{n}}}} \right\|_2^2$, respectively.

In the ISAC, the decoding order is (i) sensing and (ii) communication signals. Hence, it is assumed that  the channel gain of the sensing signal is greater than that of the communication signal, i.e, $\left\| {{\bf{h}}_{{\rm{bo}}}^{\rm{H}}{{\bf{w}}_{{\rm{bo}}}}} \right\|_2^2 < \left\| {{\bf{h}}_{{\rm{bo}}}^{\rm{H}}{{\bf{w}}_{{\rm{ta}}}}} \right\|_2^2$ and $\left\| {{\bf{h}}_{{\rm{bo}}}^{\rm{H}}{{\bf{w}}_{{\rm{bo}}}}} \right\|_2^2 < \left\| {{\bf{h}}_{{\rm{bo}}}^{\rm{H}}{{\bf{w}}_{{\rm{ta}}}}} \right\|_2^2$

The sensing SINR at Bob can be obtained as
\begin{equation}\label{Bob_sensing_SINR}
{\gamma _{{\rm{bo,ta}}}} = \frac{{\left\| {{\bf{h}}_{{\rm{bo}}}^{\rm{H}}{{\bf{w}}_{{\rm{ta}}}}} \right\|_2^2}}{{\left\| {{\bf{h}}_{{\rm{bo}}}^{\rm{H}}{{\bf{w}}_{{\rm{bo}}}}} \right\|_2^2 + \left\| {{\bf{g}}_{{\rm{bo}}}^{\rm{H}}{{\bf{o}}_{\rm{c}}}} \right\|_2^2 + \sum\nolimits_{n = 1}^N {\left\| {{\bf{g}}_{{\rm{bo}}}^{\rm{H}}{{\bf{o}}_n}} \right\|_2^2}  + {\sigma ^2}}}.
\end{equation}
To employ SIC and remove sensing signals and reserve communication signals, the following condition needs to be satisfied: $ {\gamma _{{\rm{bo,ta}}}} \ge {\gamma _{{\rm{th}}}}$,
where ${\gamma _{{\rm{th}}}}$ is the required threshold for decoding sensing signals successfully. Then, after eliminating sensing signals at Bob and accounting for the bounded errors, the worst-case SINR at Bob can be expressed as
\begin{equation}\label{Bob_communication_SINR}
\gamma _{{\rm{bo,bo}}}^{{\rm{worst}}} = \frac{{{\bf{w}}_{{\rm{bo}}}^{\rm{H}}{{\bf{H}}_{{\rm{bo,min}}}}{{\bf{w}}_{{\rm{bo}}}}}}{{{\bf{o}}_{\rm{c}}^{\rm{H}}{{\bf{G}}_{{\rm{bo,max}}}}{{\bf{o}}_{\rm{c}}} + \sum\limits_{n = 1}^N {{\bf{o}}_n^H} {{\bf{G}}_{{\rm{bo,max}}}}{{\bf{o}}_n} + {\sigma ^2}}},
\end{equation}
where ${{\bf{H}}_{{\rm{bo,min}}}} = {{\bf{h}}_{{\rm{bo,es}}}}{\bf{h}}_{{\rm{bo,es}}}^{\rm{H}} - {{{e}}_{{\rm{h}}}}({\left\| {{{\bf{h}}_{{\rm{bo,es}}}}} \right\|_2} + {{{e}}_{{\rm{h}}}}){{\bf{I}}_{M_1 \times M_1}}$ and ${{\bf{G}}_{{\rm{bo,max}}}} = {{\bf{g}}_{{\rm{bo,es}}}}{\bf{g}}_{{\rm{bo,es}}}^{\rm{H}} + {{{e}}_{{\rm{g}}}}({\left\| {{{\bf{g}}_{{\rm{bo,es}}}}} \right\|_2} + {{{e}}_{{\rm{g}}}}){\bf{I}}_{M_2 \times M_2}$.
Similarly, the sensing SINR at Eve is 
\begin{equation}\label{Eve_sensing_SINR}
{\gamma _{{\rm{e,ta}}}} = \frac{{\left\| {{\bf{h}}_{\rm{e}}^{\rm{H}}{{\bf{w}}_{{\rm{ta}}}}} \right\|_2^2}}{{\left\| {{\bf{h}}_{\rm{e}}^{\rm{H}}{{\bf{w}}_{{\rm{bo}}}}} \right\|_2^2 + \left\| {{\bf{g}}_{\rm{e}}^{\rm{H}}{{\bf{o}}_{\rm{c}}}} \right\|_2^2 + \sum\nolimits_{n = 1}^N {\left\| {{\bf{g}}_{\rm{e}}^{\rm{H}}{{\bf{o}}_n}} \right\|_2^2}  + {\sigma ^2}}}.
\end{equation}
In contrast to Bob, the Eve is constrained to be failed to decode
the target’s message, obtaining the communication signal with the sensing signal interference, and thus degradating the eavesdropping SINR, the following condition needs to be satisfied: $	{\gamma _{{\rm{bo,ta}}}} <  {\gamma _{{\rm{th}}}}$.
Therefore, the communication SINR at Eve can be respectively expressed as
 \begin{equation}\label{Eve_communication_SINR}
 {\gamma _{{\rm{e,bo}}}} = \frac{{\left\| {{{\bf{h}}_{\rm{e}}}{\bf{w}}_{{\rm{bo}}}^{\rm{H}}} \right\|_2^2}}{{\left\| {{{\bf{h}}_{\rm{e}}}{\bf{w}}_{{\rm{ta}}}^{\rm{H}}} \right\|_2^2 + \left\| {{{\bf{g}}_{\rm{e}}}{\bf{o}}_{\rm{c}}^{\rm{H}}} \right\|_2^2 + \sum\nolimits_{n = 1}^N {\left\| {{{\bf{g}}_{\rm{e}}}{\bf{o}}_n^{\rm{H}}} \right\|_2^2}  + {\sigma ^2}}}.
 \end{equation}
Thus, the worst-case security rate is given by
 \begin{equation}\label{defi_RS}
{R^{\rm worst}_{\rm{S}}} = {\left[ {{R_{{\rm{bo}}}^{\rm worst}} - {R_{\rm{e}}}} \right]^ + },
  \end{equation}
where $	{R_{{\rm{bo}}}^{\rm worst}} = {\log _2}\left( {1 + {\gamma _{{\rm{bo,bo}}}^{\rm worst}}} \right)$ and ${R_{\rm{e}}} = {\log _2}\left( {1 + {\gamma _{{\rm{e,bo}}}}} \right)$ are the worst-case legitimate rate at Bob and eavesdropping rate at Eve, and ${\left[ x \right]^ + } = \max \left\{ {x,0} \right\}$.

Meanwhile, $\rm BS_2$ provides RSMA-based downlink data transmission to $N$ users. According to RSMA, the common information $s_{\rm c}$ is decoded into a common stream first, and the private information $s_n$ is regarded as interference. Then, by applying SIC, the common stream is recoded, pre-encoded, and removed from the received signal. The private information $s_n$ of the $n$-th user is decoded into a private stream, and private informations of other users is treated as interference.
Consequently, the SINRs of the common and the private streams of the $n$-th user can be respectively expressed as
\begin{equation}\label{common_SINR}
{\gamma _{\rm{c}}} = \frac{{\left\| {{\bf{g}}_{n}^{\rm{H}}{{\bf{o}}_{\rm{c}}}} \right\|_2^2}}{{\sum\nolimits_{n = 1}^N {\left\| {{\bf{g}}_{n}^{\rm{H}}{{\bf{o}}_n}} \right\|_2^2}  + \left\| {{\bf{h}}_{n}^{\rm{H}}{{\bf{w}}_{{\rm{bo}}}}} \right\|_2^2 + \left\| {{\bf{h}}_{n}^{\rm{H}}{{\bf{w}}_{{\rm{ta}}}}} \right\|_2^2 + {\sigma ^2}}},
\end{equation}
\begin{equation}\label{private_SINR}
{\gamma _n} = \frac{{\left\| {{\bf{g}}_n^{\rm{H}}{{\bf{o}}_n}} \right\|_2^2}}{{\sum\nolimits_{j \ne n} {\left\| {{\bf{g}}_n^{\rm{H}}{{\bf{o}}_j}} \right\|_2^2}  + \left\| {{\bf{h}}_n^{\rm{H}}{{\bf{w}}_{{\rm{bo}}}}} \right\|_2^2 + \left\| {{\bf{h}}_n^{\rm{H}}{{\bf{w}}_{{\rm{ta}}}}} \right\|_2^2 + \sigma _n^2}}.
\end{equation}
The corresponding common and private stream rates can be respectively given by ${R_{\rm{c}}} = {\log _2}\left( {1 + {\gamma _{\rm{c}}}} \right)$ and ${R_{{n}}} = {\log _2}\left( {1 + {\gamma _{{n}}}} \right)$.

\subsection{Sensing Signal Model}
Since $\rm BS_1$ is fully aware of its own transmit signal ${{\bf{x}}_1}$, which is composed of the sensing signal ${{\bf{w}}_{{\rm{ta}}}}{x_{{\rm{ta}}}}$ and the communication signal ${{\bf{w}}_{{\rm{bo}}}}{x_{{\rm{bo}}}}$, ${{\bf{x}}_1}$ can be used to detect the target. Meanwhile, sensing is affected by clutter from the environment. Therefore, the echo signal received at $\rm BS_1$ can be expressed as
\begin{equation}\label{echo_signal}
\begin{array}{l} \displaystyle
	{{\bf{y}}_{{\rm{b1}}}} = \xi {{\bf{h}}_{{\rm{ta}}}}{\bf{h}}_{{\rm{ta}}}^{\rm{H}}\left( {{{\bf{w}}_{{\rm{bo}}}}{x_{{\rm{bo}}}} + {{\bf{w}}_{{\rm{ta}}}}{x_{{\rm{ta}}}}} \right)\\  \displaystyle \quad\;\;\;
	+ {{\bf{I}}_{{M_{1 \times 1}}}}{\bf{h}}_{{\rm{cl}}}^{\rm{H}}\left( {{{\bf{w}}_{{\rm{bo}}}}{x_{{\rm{bo}}}} + {{\bf{w}}_{{\rm{ta}}}}{x_{{\rm{ta}}}}} \right) + {{\bf{n}}_{{\rm{b1}}}},
\end{array}
\end{equation}
where ${{\bf{h}}_{{\rm{ta}}}} \in {\mathbb{C}^{{M_1} \times 1}}$ is the channel from $\rm BS_1$ to the target and ${{\bf{h}}_{{\rm{cl}}}} \in {\mathbb{C}^{{M_1} \times 1}}$ is the channel from the environment to $\rm BS_1$. 
Additionally, $\xi $ is the \emph{radar cross section} (RCS) coefficient with the mean square value being ${\kappa ^2}$, and ${{\bf{n}}_{{\rm{b1}}}} \sim CN\left( {0,{\sigma ^2}{{\bf{I}}_{{M_1} \times 1}}} \right)$ represents the complex random Gaussian noise at $\rm BS_1$ with a mean being 0 and a variance being $\sigma^2$. We model the clutter from the environment as complex random Gaussian noise with a mean being 0 and a variance being $\sigma^2$ \cite{ISAC_PLS_54}.

Then, a radar receiver filter vector at $\rm BS_1$ is defined as ${{\bf{a}}} \in {\mathbb{C}^{{M_1} \times 1}}$, and thus the further processed echo signal and the corresponding SCNR can be respectively obtained as
\begin{equation}\label{a_echo_signal}
\begin{array}{l} \displaystyle
	{\bf{a}}^{\rm{H}}{{\bf{y}}_{{\rm{b1}}}} = \xi {\bf{a}}^{\rm{H}}{{\bf{h}}_{{\rm{ta}}}}{\bf{h}}_{{\rm{ta}}}^{\rm{H}}\left( {{{\bf{w}}_{{\rm{bo}}}}{x_{{\rm{bo}}}} + {{\bf{w}}_{{\rm{ta}}}}{x_{{\rm{ta}}}}} \right)\\ \displaystyle \qquad\quad \,
	+ {\bf{a}}^{\rm{H}}{{\bf{I}}_{{M_1} \times 1}}{\bf{h}}_{{\rm{cl}}}^{\rm{H}}\left( {{{\bf{w}}_{{\rm{bo}}}}{x_{{\rm{bo}}}} + {{\bf{w}}_{{\rm{ta}}}}{x_{{\rm{ta}}}}} \right) + {\bf{a}}^{\rm{H}}{{\bf{n}}_{{\rm{b1}}}},
\end{array}
\end{equation}
and
\begin{equation}\label{echo_SCNR}
\gamma _{{\rm{b1}}}= \frac{{{\kappa ^2}{{\bf{a}}^{\rm{H}}}{{\bf{H}}_{{\rm{ta-b1}}}}({{\bf{W}}_{{\rm{bo}}}} + {{\bf{W}}_{{\rm{ta}}}}){\bf{H}}_{{\rm{ta-b1}}}^{\rm{H}}{\bf{a}}}}{{2{\sigma ^2}{{\bf{a}}^{\rm{H}}}{\bf{a}}}}.
\end{equation}
Likewise, to successfully decode the echo signal at $\rm BS_1$, the related constraint $\gamma _{{\rm{b1}}} \ge {\gamma _{{\rm{th}}}}$ should be satisfied.

\section{Optimization Formulation and Solution}
We aim at jointly enhancing the system's sensing, communication, security, and energy efficiency. Given the intractability of closed-form solutions for this complex problem, we focus on developing an efficient iterative algorithm.  Our novelty lies in the system modeling, problem formulation, and algorithm development for the joint optimization in a heterogeneous ISAC network. The mathematical derivations we presented are essential for developing the computationally efficient alternating optimization algorithm.
To this end, we design a joint BF optimization scheme to reduce the interference exerted on Bob, suppress Eve's eavesdropping, and improve the sensing accuracy and energy utilization based on RSMA and the constructed green interference. Specifically, our objective is to maximize the SEE of the entire heterogeneous ISAC. The optimization objective, SEE, is related to the legitimate rate of Bob $R_{\rm bo}$, the eavesdropping rate of the Eve $R_{\rm e}$, the power consumptions of $\rm BS_1$, $\rm BS_2$, and the circuit, $P_1$, $P_2$, and $P_0$, respectively. Optimization of $\textbf{a}$ reduces $P_1$ while $\rm BS_1$ can successfully decode the echo signal. Optimization of $\textbf{w}_{\rm bo}$ leads to a larger $R_{\rm bo}$, optimizing $\textbf{w}_{\rm ta}$ can reduce $P_1$ and $R_{\rm e}$, $\textbf{o}_{\rm c}$ and $\textbf{o}_{n}$ are optimized to increase the multicast communication rate and reduce $R_{\rm e}$. Consequently, the optimization problem can be mathematically expressed as
\begin{subequations} \label{P1}
	\begin{align}
		&{\bf{P1}}. \;\mathop {\max }\limits_{{\bf{a}},{{\bf{w}}_{{\rm{bo}}}},{{\bf{w}}_{{\rm{ta}}}},{{\bf{o}}_{\rm{c}}},{{\bf{o}}_n}} \frac{{R_{\rm{S}}^{{\rm{worst}}}}}{{{P_1} + {P_2} + {P_0}}}\\
		& {\rm{s}}{\rm{.t.}}\quad {\rm{    }}\left\| {{{\bf{w}}_{{\rm{bo}}}}} \right\|_2^2 + \left\| {{{\bf{w}}_{{\rm{ta}}}}} \right\|_2^2 = {P_1} \le {P_{{\rm{1,max}}}},\\ \label{P1_C5}
		&  \quad\quad\,\,{\rm{    }}\left\| {{{\bf{o}}_{\rm{c}}}} \right\|_2^2 + \sum\nolimits_{n = 1}^N {\left\| {{{\bf{o}}_n}} \right\|_2^2}  = {P_2} \le {P_{{\rm{2,max}}}}, 	\\ \label{P1_C5-2}		
		&\quad\quad\,\, {R_{\rm{S}}^{\rm worst}} \ge {I_{\rm{S}}}\, \& \, {R_{\rm{c}}} \ge {I_{\rm{c}}}\, \&\, {R_n} \ge {I_{\rm p}},\\ \label{P1_C2}
		&\quad\quad\,\,  \left\| {{\bf{h}}_{{\rm{bo}}}^{\rm{H}}{{\bf{w}}_{{\rm{bo}}}}} \right\|_2^2 \le \left\| {{\bf{h}}_{{\rm{bo}}}^{\rm{H}}{{\bf{w}}_{{\rm{ta}}}}} \right\|_2^2,\\ \label{P1_C3}
		&\quad\quad\,\, \left\| {{\bf{h}}_{\rm{e}}^{\rm{H}}{{\bf{w}}_{{\rm{bo}}}}} \right\|_2^2 \le \left\| {{\bf{h}}_{\rm{e}}^{\rm{H}}{{\bf{w}}_{{\rm{ta}}}}} \right\|_2^2,\\ \label{P1_C3-2}
		&\quad\quad\,\, {\rm{     }}{\gamma _{{\rm{bo}},{\rm{ta}}}} \ge {\gamma _{{\rm{th}}}} > {\mkern 1mu} {\gamma _{{\rm{e}},{\rm{ta}}}},	\gamma _{{\rm{b1}}}\ge {\gamma _{{\rm{th}}}}, \, \gamma _{{\rm{bo,bo}}}^{{\rm{worst}}} \ge {\gamma _{{\rm{th}}}},
	\end{align}	
\end{subequations}
where (\ref{P1}a) is the optimization objective on SEE maximization. (\ref{P1}b) and (\ref{P1}c) are constraints of $\rm BS_1$ and $\rm BS_2$ power consumptions, $P_{\rm 1,max}$ and $P_{\rm 2,max}$ are maximums of $P_1$ and $P_2$, $P_0$ is constant circuit power consumption. (\ref{P1}d) are constraints on security rate, common stream rate, and private stream rate, respectively. $I_{\rm S}$, $I_{\rm c}$, and $I_{\rm p}$ are thresholds of security rate, common stream rate, and private stream rate. (\ref{P1}e) and (\ref{P1}f) are constraints on the decoding sequence of the sensing signal and the communication signal. (\ref{P1}g) are constraints that Bob and $\rm BS_1$ successfully decode sensing signals and Eve fails to decode its received sensing signal, deemed as interference for Eve, as well as the related worst-case constraints on thresholds. What's more, we define ${P_{{\rm{sum}}}} = {P_1} + {P_2} + {P_0}$.

It can be observed that there is a deep coupling relationship between the variables to be optimized in \textbf{P1} and the optimization objective and corresponding constraints, which brings notable challenges for effectively and reliably solving \textbf{P1}. To this end, we present an alternative iterative optimization algorithm based on Taylor series expansion, MM, SDP, and SCA, respectively.
In particular, we decompose the original non-convex and intractable targeted optimization \textbf{P1} into three sub-optimization problems. The first sub-optimization problem focuses on the optimization of the echo signal BF vector at $\rm BS_1$. The second one considers the optimization of the transmit BF vector at $\rm BS_1$. The last one is the optimization of the transmit BF vector at $\rm BS_2$.
In what follows, it is defined ${{\bf{X}}_i} = {{\bf{x}}_i}{\bf{x}}_i^{\rm{H}}$.
The fractional form of the objective and the coupled variables {$\textbf{a}$, $\textbf{w}_{\rm bo}$, $\textbf{w}_{\rm ta}$, $\textbf{o}_{\rm c}$, $\textbf{o}_n$} render problem (\ref{P1}) intractable. Therefore, we split the original problem into three sub-problems, and propose an alternating optimization scheme to solve it.

\subsection{Optimization of the Echo Signal BF Vector at $\rm BS_1$}
To optimize the echo signal BF vector at $\rm BS_1$ $\textbf{a}$, we keep $\textbf{w}_{\rm bo}$, $\textbf{w}_{\rm ta}$, $\textbf{o}_{\rm c}$, and $\textbf{o}_n$ fixed, and we present the equivalent form of the constraint  (\ref{P1}h) as 
\begin{subequations} \label{PA1}
	\begin{align}
		&{\rm{        }}{\bf{P2.1}}{\rm{.}}\; \mathop {\max }\limits_{\bf{a}} \frac{{{\kappa ^2}}}{{2{\sigma ^2}}}\frac{{{{\bf{a}}^{\rm{H}}}\left( {{{\bf{H}}_{{\rm{ta - b1}}}}{{\bf{w}}_{{\rm{b1}}}}{\bf{w}}_{{\rm{b1}}}^{\rm{H}}{\bf{H}}_{{\rm{ta - b1}}}^{\rm{H}}} \right){\bf{a}}}}{{{{\bf{a}}^{\rm{H}}}{\bf{a}}}}\\
		&{\rm{s}}{\rm{.t.}}\quad \, (\ref{P1}\rm b) \, \mbox{-}  \, (\ref{P1}\rm g).
	\end{align}	
\end{subequations}
where ${{\bf{H}}_{{\rm{ta,b1}}}} = {{\bf{h}}_{{\rm{ta}}}}{\bf{h}}_{{\rm{ta}}}^{\rm{H}} \in {\mathbb{C}^{{M_1} \times {M_1}}}$ and ${{\bf{w}}_{{\rm{b1}}}} = \left[ {{{\bf{w}}_{{\rm{bo}}}},{{\bf{w}}_{{\rm{ta}}}}} \right] \in {\mathbb{C}^{{M_1} \times 2}}$.

Herein, $\textbf{S}$ is a Hermitian matrix, which can be expressed as
\begin{equation}\label{Hermitian_matrix}
{{{\bf{S}}^{\rm{H}}} = {{\left[ {\frac{{{\kappa ^2}{{\bf{h}}_{{\rm{ta}}}}{\bf{h}}_{{\rm{ta}}}^{\rm{H}}\left( {{{\bf{w}}_{{\rm{bo}}}}{\bf{w}}_{{\rm{bo}}}^{\rm{H}} + {{\bf{w}}_{{\rm{ta}}}}{\bf{w}}_{{\rm{ta}}}^{\rm{H}}} \right){{\bf{h}}_{{\rm{ta}}}}{\bf{h}}_{{\rm{ta}}}^{\rm{H}}}}{{2{\sigma ^2}}}} \right]}^{\rm{H}}} = {\bf{S}}}.
\end{equation}
We diagonalize the Hermitian matrix $\textbf{S}$, define the eigenvalue diagonal matrix of $\textbf{S}$ as ${\bf{\tilde S}} = {\mathop{\rm diag}\nolimits} \left\{ {{\lambda _1},{\lambda _2}, \cdots ,{\lambda _{{M_1}}}} \right\} \in {\mathbb{C}^{{M_1} \times {M_1}}}$, define the eigenvector matrix of $\textbf{S}$ as ${\bf{Q}} = \left\{ {{{\bf{q}}_1},{{\bf{q}}_2}, \cdots ,{{\bf{q}}_{{M_1}}}} \right\} \in {\mathbb{C}^{{M_1} \times {M_1}}}$  with ${\bf{Q}}{{\bf{Q}}^{\rm{H}}} = {{\bf{I}}_{{M_1} \times {M_1}}}$. Then, (\ref{PA1}) can be rewritten as
\begin{subequations} \label{PA2}
	\begin{align}
		&{\rm{        }}{\bf{P2.2}}{\rm{.}}\;{\mathop {\max }\limits_{\bf{a}} \frac{{{{\bf{a}}^{\rm{H}}}{\bf{Q\tilde S}}{{\bf{Q}}^{\rm{H}}}{\bf{a}}}}{{{{\bf{a}}^{\rm{H}}}{\bf{Q}}{{\bf{Q}}^{\rm{H}}}{\bf{a}}}}}  \\
		&{\rm{s}}{\rm{.t.}}\quad \, (\ref{P1}\rm b) \, \mbox{-} \, (\ref{P1}\rm g).
	\end{align}	
\end{subequations}
Besides, let ${\bf{\tilde Q}} = {{\bf{Q}}^{\rm{H}}}{\bf{a}} = {\left\{ {{{{\bf{\tilde q}}}_1},{{{\bf{\tilde q}}}_2}, \cdots ,{{{\bf{\tilde q}}}_{{M_1}}}} \right\}^{\rm{H}}}$. By applying ${\bf{\tilde Q}}$ to  (\ref{PA2}), we obtain
\begin{subequations} \label{PA3}
	\begin{align}
		&{\rm{        }}{\bf{P2.3}}{\rm{.}}\;{\mathop {\max }\limits_{\bf{a}} \frac{{\sum\nolimits_{j = 1}^{{M_1}} {{\lambda _j}{\bf{\tilde q}}_j^2} }}{{\sum\nolimits_{j = 1}^{{M_1}} {{\bf{\tilde q}}_j^2} }}}\\
		&{\rm{s}}{\rm{.t.}}\quad \, (\ref{P1}\rm b) \, \mbox{-} \, (\ref{P1}\rm g).
	\end{align}	
\end{subequations}
Then, after defining ${\lambda _{\max }} = \max \left\{ {{\lambda _1},{\lambda _2}, \cdots ,{\lambda _{{M_1}}}} \right\}$ and ${\lambda _{\min }} = \min \left\{ {{\lambda _1},{\lambda _2}, \cdots ,{\lambda _{{M_1}}}} \right\}$, we obtain
\begin{equation}\label{PA3_max}
\frac{{\sum\nolimits_{j = 1}^{{M_1}} {{\lambda _j}{\bf{\tilde q}}_j^2} }}{{\sum\nolimits_{j = 1}^{{M_1}} {{\bf{\tilde q}}_j^2} }} \le \frac{{\sum\nolimits_{j = 1}^{{M_1}} {{\lambda _{\max }}{\bf{\tilde q}}_j^2} }}{{\sum\nolimits_{j = 1}^{{M_1}} {{\bf{\tilde q}}_j^2} }} = \frac{{{\lambda _{\max }}\sum\nolimits_{j = 1}^{{M_1}} {{\bf{\tilde q}}_j^2} }}{{\sum\nolimits_{j = 1}^{{M_1}} {{\bf{\tilde q}}_j^2} }} = {\lambda _{\max }},
\end{equation}
\begin{equation}\label{PA3_min}
\frac{{\sum\nolimits_{j = 1}^{{M_1}} {{\lambda _j}{\bf{\tilde q}}_j^2} }}{{\sum\nolimits_{j = 1}^{{M_1}} {{\bf{\tilde q}}_j^2} }} \ge \frac{{\sum\nolimits_{j = 1}^{{M_1}} {{\lambda _{\min }}{\bf{\tilde q}}_j^2} }}{{\sum\nolimits_{j = 1}^{{M_1}} {{\bf{\tilde q}}_j^2} }} = \frac{{{\lambda _{\min }}\sum\nolimits_{j = 1}^{{M_1}} {{\bf{\tilde q}}_j^2} }}{{\sum\nolimits_{j = 1}^{{M_1}} {{\bf{\tilde q}}_j^2} }} = {\lambda _{\min }}.
\end{equation}
Finally, (\ref{PA1}) can be simplified as
\begin{subequations} \label{PA4}
	\begin{align}
		&{\rm{        }}{\bf{P2}}{\rm{.}}\; \mathop {\max }\limits_{\bf{a}} \frac{{{{\bf{a}}^{\rm{H}}}{\bf{Sa}}}}{{{{\bf{a}}^{\rm{H}}}{\bf{a}}}} = {\lambda _{\max }}\\
		&{\rm{s}}{\rm{.t.}}\quad \, (\ref{P1}\rm b) \, \mbox{-}  \, (\ref{P1}\rm g).
	\end{align}	
\end{subequations}
In this case, the optimization variable $\textbf{a}$ is equivalent to the eigenvector corresponding to the maximum eigenvalue of \textbf{S}, i.e., ${\bf{Sa}} = {\lambda _{\max }}{\bf{a}}$.
We use QR decomposition to compute the maximum of the target optimization, i.e., the maximum eigenvalue of \textbf{S}, and the corresponding eigenvector, i.e., the optimization variable $\textbf{a}$. The corresponding calculation process is provided in Algorithm \ref{algorithm_1}.

\subsection{Optimization of the Transmit BF Vectors at $\rm BS_1$}
For fixed $\textbf{a}$, $\textbf{o}_{\rm c}$, and $\textbf{o}_n$, we optimize $\textbf{w}_{\rm bo}$ and $\textbf{w}_{\rm ta}$ in \textbf{P1} jointly. 
The security rate in \textbf{P1} can be rewritten as
\begin{equation}\label{security_rewritten}
	\begin{array}{l} \displaystyle
		{R_{\rm{S}}} = {R_{\rm{1}}} - {R_{\rm{2}}} + {R_{\rm{3}}} - {R_{\rm{4}}}\\ \displaystyle \quad\;\;
		= {\log _2}\left( {{\mathop{\rm Tr}\nolimits} \left( {{{\bf{H}}_{{\rm{bo,min}}}}{{\bf{W}}_{{\rm{bo}}}}} \right) + {\alpha _{{\rm{bo}}}}} \right)\\ \displaystyle\quad\;\;
		- {\log _2}\left( {{\mathop{\rm Tr}\nolimits} \left( {{{\bf{H}}_{\rm{e}}}{{\bf{W}}_{{\rm{bo}}}}} \right) + {\mathop{\rm Tr}\nolimits} \left( {{{\bf{H}}_{\rm{e}}}{{\bf{W}}_{{\rm{ta}}}}} \right) + {\alpha _{\rm{e}}}} \right)\\ \displaystyle \quad\;\;
		+ {\log _2}\left( {{\mathop{\rm Tr}\nolimits} \left( {{{\bf{H}}_{\rm{e}}}{{\bf{W}}_{{\rm{ta}}}}} \right) + {\alpha _{\rm{e}}}} \right) - {\log _2}\left( {{\alpha _{{\rm{bo}}}}} \right),
	\end{array}
\end{equation}
where ${\alpha _{{\rm{bo}}}} = {\mathop{\rm Tr}\nolimits} \left( {{{\bf{G}}_{{\rm{bo,max}}}}{{\bf{O}}_{\rm{c}}}} \right) + \sum\nolimits_{n = 1}^N {{\mathop{\rm Tr}\nolimits} \left( {{{\bf{G}}_{{\rm{bo,max}}}}{{\bf{O}}_n}} \right)}  + {\sigma ^2}$, and ${\alpha _{\rm{e}}} = {\mathop{\rm Tr}\nolimits} \left( {{{\bf{G}}_{\rm{e}}}{{\bf{O}}_{\rm{c}}}} \right) + \sum\nolimits_{n = 1}^N {{\mathop{\rm Tr}\nolimits} \left( {{{\bf{G}}_{\rm{e}}}{{\bf{O}}_n}} \right)}  + {\sigma ^2}$ are constants.

\begin{algorithm}[t]
	\caption{Optimization of echo signal BF vector at $\rm BS_1$.}
	\label{algorithm_1}
	{\algorithmicrequire} Given $\textbf{w}_{\rm bo}$, $\textbf{w}_{\rm ta}$, $\textbf{o}_{\rm c}$, and $\textbf{o}_n$, and related variables in the defined models.
	
	{\algorithmicensure} Optimal echo signal BF vector $\textbf{a}$.
	\begin{algorithmic}[1] 
		\State Set the mean square value of RCS $\kappa$, iteration threshold $\delta$, $\varepsilon=0$, $\lambda_{0}=0$, and $\lambda_{-1}=0$;
		\While {all constraints in (\ref{PA4}\rm b) are satisfied}
		\Repeat 
		\State Solve the optimization objective in (\ref{PA4}) to obtain 
		\Statex \qquad \;\,\, the maximum eigenvalue $\lambda_\varepsilon$;
		\State Update ${{\bf{a}}_\varepsilon } = {\mathop{\rm QR}\nolimits} \left( {{\bf{S}},{\lambda _\varepsilon }} \right)$; 
		\State $\varepsilon  = \varepsilon  + 1$;
		\Until {${\lambda _\varepsilon } - {\lambda _{\varepsilon  - 1}} \le {\delta}$};
		\EndWhile
		\State Calculate optimal echo signal BF vector $\textbf{a}$ at $\rm BS_1$ with QR decomposition based on maximum eigenvalue, $\lambda_{\rm max}$;
	\end{algorithmic}
\end{algorithm}

\begin{algorithm}[htbp]
	\caption{Optimization of the transmit BF vectors at $\rm BS_1$.}
	\label{algorithm_2}
	{\algorithmicrequire} Given $\textbf{a}$, $\textbf{o}_{\rm c}$, and $\textbf{o}_n$, and related variables in the defined models.
	
	{\algorithmicensure} Optimal transmit signal BF vectors $\textbf{w}_{\rm bo}$ and $\textbf{w}_{\rm ta}$.
	\begin{algorithmic}[1] 
		\State Set auxiliary variables $r_{\rm a}$ and $s_{\rm a}$, iteration threshold $\delta$, $\varepsilon=0$, ${\rm SEE}_{0}=0$, and ${\rm SEE}_{-1}=0$;
		\While {all constraints in (\ref{PBB-3}\rm b) are satisfied}
		\State Set $i = 0$, and initialize ${r_{{\rm a},0,i }}$ and ${s_{{\rm a},0,i }}$;
		\Repeat 
		\State $i  = i  + 1$;
		\State Solve the optimization objective in (\ref{PBB-3}) to obtain 
		\Statex \qquad \;\,\, maximum ${\rm SEE}=\left( {2{r_{\rm{a}}}\sqrt {{R^{\rm worst}_{{\rm{S}}}}\left( r \right)}  - r_{\rm{a}}^2{P_{{\rm{sum}}}}} \right)$;
		\State Update $\textbf{W}_{{\rm bo},{\varepsilon}}$, $\textbf{W}_{{\rm ta},{\varepsilon}}$, $r_{{\rm a, 0},i}$, and $s_{{\rm a, 0},i}$;
		\Until {${\rm{SE}}{{\rm{E}}_\varepsilon } - {\rm{SE}}{{\rm{E}}_{\varepsilon  - 1}} \le {\delta }$};
		\State Obtain solutions $\textbf{W}_{\rm bo}$ and $\textbf{W}_{\rm ta}$;
		\State Set ${{\bf{W}}_{{\rm{bo}},\varepsilon  + 1}}{\rm{ = }}{{\bf{W}}_{\rm{bo}}}$ and ${{\bf{W}}_{{{\rm ta}},\varepsilon  + 1}}{\rm{ = }}{{\bf{W}}_{{\rm ta}}}$;
		\While {${{\bf{W}}_{{\rm{bo}},\varepsilon  + 1}} \approx {{\bf{W}}_{{\rm{bo}},\varepsilon }}$, and ${{\bf{O}}_{{{n}},\varepsilon  + 1}} \approx {{\bf{O}}_{{{n}},\varepsilon }}$  are not \indent \indent  satisfied}
		\State $\varepsilon  = \varepsilon  + 1$;
		\EndWhile
		\EndWhile
		\State Employ \emph{singular value decomposition} (SVD) to $\textbf{W}_{\rm bo, \varepsilon}$ and $\textbf{W}_{\rm ta,\varepsilon}$, and then the transmit signal BF vectors $\textbf{w}_{\rm bo}$ and $\textbf{w}_{\rm ta}$ are ascertained, respectively;
	\end{algorithmic}
\end{algorithm} 

According to the MM, the following inequality is satisfied: ${\log _2}\left( z \right) \le \frac{{z - {z_0}}}{{{z_0}}} + {\log _2}\left( {{z_0}} \right)$, where $z_0$ is a specific value of variable $z$. Therefore, the second item and the third item in (\ref{security_rewritten}) can be upper-bounded as
\begin{equation}\label{2nd_UB}
\begin{array}{l} \displaystyle
	{R_{\rm{2}}} \le \frac{{{\mathop{\rm Tr}\nolimits} \left( {{{\bf{H}}_{\rm{e}}}{{\bf{W}}_{{\rm{bo}}}}} \right) - {\mathop{\rm Tr}\nolimits} \left( {{{\bf{H}}_{\rm{e}}}{{\bf{W}}_{{\rm{bo,}}i}}} \right)}}{{{\mathop{\rm Tr}\nolimits} \left( {{{\bf{H}}_{\rm{e}}}{{\bf{W}}_{{\rm{bo,}}i}}} \right) + {\mathop{\rm Tr}\nolimits} \left( {{{\bf{H}}_{\rm{e}}}{{\bf{W}}_{{\rm{ta}}}}} \right) + {\alpha _{\rm{e}}}}}\\  \displaystyle \quad\,\,\,
	+ \frac{{{\mathop{\rm Tr}\nolimits} \left( {{{\bf{H}}_{\rm{e}}}{{\bf{W}}_{{\rm{ta}}}}} \right) - {\mathop{\rm Tr}\nolimits} \left( {{{\bf{H}}_{\rm{e}}}{{\bf{W}}_{{\rm{ta,}}j}}} \right)}}{{{\mathop{\rm Tr}\nolimits} \left( {{{\bf{H}}_{\rm{e}}}{{\bf{W}}_{{\rm{bo,}}i}}} \right) + {\mathop{\rm Tr}\nolimits} \left( {{{\bf{H}}_{\rm{e}}}{{\bf{W}}_{{\rm{ta,}}j}}} \right) + {\alpha _{\rm{e}}}}}\\  \displaystyle \quad\,\,\,
	+ {\log _2}\left( {{\mathop{\rm Tr}\nolimits} \left( {{{\bf{H}}_{\rm{e}}}{{\bf{W}}_{{\rm{bo,}}i}}} \right) + {\mathop{\rm Tr}\nolimits} \left( {{{\bf{H}}_{\rm{e}}}{{\bf{W}}_{{\rm{ta,}}j}}} \right) + {\alpha _{\rm{e}}}} \right),
\end{array}
\end{equation}
and
\begin{equation}\label{3rd_UB}
\begin{array}{l} \displaystyle
	{R_{\rm{3}}} \le \frac{{{\mathop{\rm Tr}\nolimits} \left( {{{\bf{H}}_{\rm{e}}}{{\bf{W}}_{{\rm{ta}}}}} \right) - {\mathop{\rm Tr}\nolimits} \left( {{{\bf{H}}_{\rm{e}}}{{\bf{W}}_{{\rm{ta,}}i}}} \right)}}{{{\mathop{\rm Tr}\nolimits} \left( {{{\bf{H}}_{\rm{e}}}{{\bf{W}}_{{\rm{ta,}}i}}} \right) + {\alpha _{\rm{e}}}}}\\ \displaystyle \quad\,\,\,
	+ {\log _2}\left( {{\mathop{\rm Tr}\nolimits} \left( {{{\bf{H}}_{\rm{e}}}{{\bf{W}}_{{\rm{ta,}}i}}} \right) + {\alpha _{\rm{e}}}} \right),
\end{array}
\end{equation}
where ${{{\bf{W}}_{{\rm{bo,}}i}}}$ is the $i$-th iteration result of ${{{\bf{W}}_{{\rm{bo}}}}}$, and ${{{\bf{W}}_{{\rm{ta,}}j}}}$ is the $j$-th iteration result of ${{{\bf{W}}_{{\rm{ta}}}}}$.
Based on (\ref{BS1_Eve_bo_uncertainty}) and (\ref{BS1_Eve_ta_uncertainty}), we obtain the minimum of $R_2$ in (\ref{R2_minimum}) and the maximum of $R_3$ in (\ref{R3_maximum}) at the top of this page, where ${{\bf{H}}_{{\rm{e,max}}}} = {{\bf{H}}_{{\rm{es}}}} + {\rm{e}}_{\rm{h,UB}}{{\bf{I}}_{{M_1} \times {M_1}}}$ and ${{\bf{H}}_{{\rm{e,min}}}} = {{\bf{H}}_{{\rm{es}}}} - {\rm{e}}_{\rm{h,UB}}{{\bf{I}}_{{M_1} \times {M_1}}}$.
\begin{figure*}
	\begin{equation}
		\label{R2_minimum} \displaystyle
\begin{array}{l} \displaystyle
	{R_{{\rm{2}},{\rm{min}}}} = \frac{{{\rm{Tr}}\left( {{{\bf{H}}_{{\rm{e}},{\rm{min}}}}{{\bf{W}}_{{\rm{bo}}}}} \right) - {\rm{Tr}}\left( {{{\bf{H}}_{{\rm{e}},{\rm{max}}}}{{\bf{W}}_{{\rm{bo}},i}}} \right)}}{{{\rm{Tr}}\left( {{{\bf{H}}_{{\rm{e}},{\rm{max}}}}{{\bf{W}}_{{\rm{bo}},i}}} \right) + {\rm{Tr}}\left( {{{\bf{H}}_{{\rm{e}},{\rm{max}}}}{{\bf{W}}_{{\rm{ta}}}}} \right) + {\alpha _{\rm{e}}}}} + \frac{{{\rm{Tr}}\left( {{{\bf{H}}_{{\rm{e}},{\rm{min}}}}{{\bf{W}}_{{\rm{ta}}}}} \right) - {\rm{Tr}}\left( {{{\bf{H}}_{{\rm{e}},{\rm{max}}}}{{\bf{W}}_{{\rm{ta}},j}}} \right)}}{{{\rm{Tr}}\left( {{{\bf{H}}_{{\rm{e}},{\rm{max}}}}{{\bf{W}}_{{\rm{bo}},i}}} \right) + {\rm{Tr}}\left( {{{\bf{H}}_{{\rm{e}},{\rm{max}}}}{{\bf{W}}_{{\rm{ta}},j}}} \right) + {\alpha _{\rm{e}}}}}\\ \displaystyle \qquad\quad
	+ {\log _2}\left( {{\rm{Tr}}\left( {{{\bf{H}}_{{\rm{e}},{\rm{min}}}}{{\bf{W}}_{{\rm{bo}},i}}} \right) + {\rm{Tr}}\left( {{{\bf{H}}_{{\rm{e}},{\rm{min}}}}{{\bf{W}}_{{\rm{ta}},j}}} \right) + {\alpha _{\rm{e}}}} \right),
\end{array}
	\end{equation}
\end{figure*}
\begin{figure*}
	\begin{equation}
		\label{R3_maximum} \displaystyle
	{R_{{\rm{3}},{\rm{max}}}} = \frac{{{\rm{Tr}}\left( {{{\bf{H}}_{{\rm{e}},{\rm{max}}}}{{\bf{W}}_{{\rm{ta}}}}} \right) - {\rm{Tr}}\left( {{{\bf{H}}_{{\rm{e}},{\rm{min}}}}{{\bf{W}}_{{\rm{ta}},i}}} \right)}}{{{\rm{Tr}}\left( {{{\bf{H}}_{{\rm{e}},{\rm{min}}}}{{\bf{W}}_{{\rm{ta}},i}}} \right) + {\alpha _{\rm{e}}}}} + {\log _2}\left( {{\rm{Tr}}\left( {{{\bf{H}}_{{\rm{e}},{\rm{max}}}}{{\bf{W}}_{{\rm{ta}},i}}} \right) + {\alpha _{\rm{e}}}} \right).
	\end{equation}
	\hrulefill
\end{figure*}

The worst-case $R_{\rm S}$ can be expressed as ${R^{\rm worst}_{{\rm{S}} }} = {R_1} - {R_{{\rm{2,min}}}} + {R_{{\rm{3,max}}}} - {R_{{\rm{4}}}}$.
Then, \textbf{P1} is transformed into
\begin{subequations} \label{PBBB1}
	\begin{align}
		&{\rm{        }}{\bf{P3.1}}{\rm{.}}\;\mathop {\max }\limits_{{{\bf{w}}_{{\rm{bo}}}},{{\bf{w}}_{{\rm{ta}}}}} \frac{{{R^{\rm worst}_{{\rm{S}}}}\left( {{{\bf{w}}_{{\rm{bo}}}},{{\bf{w}}_{{\rm{ta}}}}} \right)}}{{{P_1}\left( {{{\bf{w}}_{{\rm{bo}}}},{{\bf{w}}_{{\rm{ta}}}}} \right) + {P_2} + {P_0}}}\\
		&{\rm{s}}{\rm{.t.}}\quad \, (\ref{P1}\rm b) \, \mbox{-} \, (\ref{P1}\rm h).
	\end{align}	
\end{subequations}
As shown in (\ref{PBBB1}a), we need to increase $R^{\rm worst}_{\rm S}$ and decrease $P_1$ to maximize SEE through optimization of $\textbf{w}_{\rm bo}$ and $\textbf{w}_{\rm ta}$. Given optimization function in (\ref{PBBB1}) is fractional and intractable, with a slack variable $r$ and an auxiliary variable $r_{\rm a}$ introduced, the equivalent sub-optimization of (\ref{PBBB1}) is given by
\begin{subequations} \label{PBB-2}
	\begin{align}
		&{\rm{        }}{\bf{P3.2}}{\rm{.}}\,\begin{array}{*{20}{l}}
			{\mathop {\max }\limits_{{{\bf{w}}_{{\rm{bo}}}},{{\bf{w}}_{{\rm{ta}}}},{r_{\rm{a}}}} 2{r_{\rm{a}}}\sqrt {{R^{\rm worst}_{{\rm{S}}}}}  - r_{\rm{a}}^2\left( {{P_1} + {P_2} + {P_0}} \right)}
		\end{array}\\
		&{\rm{s}}{\rm{.t.}}\quad \, (\ref{P1}\rm b) \, \mbox{-} \, (\ref{P1}\rm h), \\
		& \quad \quad\,\,  \, {R^{\rm worst}_{\rm{S}}} \ge {\log _2}\left( {1 + r} \right).
	\end{align}	
\end{subequations}
Since (\ref{PBB-2}\rm c) is non-convex, we introduce an auxiliary variable $s_{\rm a}$ to transform inequality constraints into equality ones. Then, (\ref{PBB-2}\rm c) can be rewritten as
\begin{equation}\label{PB2_C2}
2{s_{\rm{a}}}\sqrt {1 + {\gamma^{\rm worst}_{{\rm{bo}},{\rm{bo}}}}}  - s_{\rm{a}}^2\left( {1 + {\gamma _{{\rm{e}},{\rm{bo}}}}} \right) \ge r.
\end{equation}

Let us define the $i$-th iteration of ${{\bf{w}}_{\rm{bo}}}$ and ${{\bf{w}}_{\rm ta}}$ are defined as ${{\bf{w}}_{{\rm{bo}},i}}$ and ${{\bf{w}}_{{\rm ta},i}}$; then corresponding approximate matrices are given by
\begin{equation}\label{w_bo_approximation}
	{{\bf{W}}_{{\rm{bo,}}i}} = {{\bf{w}}_{{\rm{bo,}}i}}{\bf{w}}_{\rm{bo}}^{\rm{H}} + {{\bf{w}}_{\rm{bo}}}{\bf{w}}_{{\rm{bo,}}i}^{\rm{H}} - {{\bf{w}}_{{\rm{bo,}}i}}{\bf{w}}_{{\rm{bo,}}i}^{\rm{H}},
\end{equation}
\begin{equation}\label{w_ta_approximation}
	{{\bf{W}}_{{\rm{ta,}}i}} = {{\bf{w}}_{{\rm{ta,}}i}}{\bf{w}}_{\rm{ta}}^{\rm{H}} + {{\bf{w}}_{\rm{ta}}}{\bf{w}}_{{\rm{ta,}}i}^{\rm{H}} - {{\bf{w}}_{{\rm{ta,}}i}}{\bf{w}}_{{\rm{ta,}}i}^{\rm{H}}.
\end{equation}
Finally, the non-convex targeted optimization \textbf{P1} can be rewritten in a convex form as follows
\begin{subequations} \label{PBB-3}
	\begin{align}
		&{\rm{        }}{\bf{P3}}{\rm{.}}\;\mathop {\max }\limits_{{{\bf{w}}_{{\rm{bo}}}},{{\bf{w}}_{{\rm{ta}}}},{r_{\rm a}},{s_{\rm a}}} 2{r_{\rm{a}}}\sqrt {{R^{\rm worst}_{{\rm{S}} }}\left( r \right)}  - r_{\rm{a}}^2{P_{{\rm{sum}}}}\\
		&{\rm{s}}{\rm{.t.}}\quad \, (\ref{P1}\rm b) \, \mbox{-} \, (\ref{P1}\rm h), and \, (\ref{PB2_C2}).
	\end{align}	
\end{subequations}
Calculation process is summarized in Algorithm \ref{algorithm_2}.

\subsection{Optimization of the Transmit BF Vector at $\rm BS_2$}
For the third sub-optimization, we need to increase $R^{\rm worst}_{\rm S}$ and decrease $P_2$ to maximize SEE through optimization of $\textbf{o}_{\rm c}$ and $\textbf{o}_n$.
For given $\textbf{a}$, $\textbf{w}_{\rm bo}$, and $\textbf{w}_{\rm ta}$, and then $\textbf{o}_{\rm c}$ and $\textbf{o}_n$ are optimized  in \textbf{P1}. For the sake of simplicity, we define: ${{\bf{\tilde H}}_{{\rm{ta - b1}}}} = {{\bf{H}}_{{\rm{ta,b1}}}}{\bf{H}}_{{\rm{ta,b1}}}^{\rm{H}}$, ${{\bf{H}}_{{\rm{bo}}}} = {{\bf{h}}_{{\rm{bo}}}}{\bf{h}}_{{\rm{bo}}}^{\rm{H}}$, ${{\bf{H}}_{\rm{e}}} = {{\bf{h}}_{\rm{e}}}{\bf{h}}_{\rm{e}}^{\rm{H}}$, ${{\bf{W}}_{{\rm{bo}}}} = {{\bf{w}}_{{\rm{bo}}}}{\bf{w}}_{{\rm{bo}}}^{\rm{H}}$, and ${{\bf{W}}_{{\rm{ta}}}} = {{\bf{w}}_{{\rm{ta}}}}{\bf{w}}_{{\rm{ta}}}^{\rm{H}}$, respectively.
Then, we hold that
\begin{equation}\label{PB_C1}
\mathop {\max }\limits_{{{\bf{h}}_{{\rm{es}}}}} {\rm{Tr}}\left( {{{\bf{H}}_{\rm{e}}}{{\bf{W}}_{{\rm{bo}}}}} \right) = {\rm{Tr}}\left[ {\left( {{{\bf{H}}_{{\rm{es}}}} + {\rm{e}}_{\rm{h,UB}}{{\bf{I}}_{{M_1} \times {M_1}}}} \right){{\bf{W}}_{{\rm{bo}}}}} \right],
\end{equation}
\begin{equation}\label{PB_C2}
	\mathop {\min }\limits_{{{\bf{h}}_{{\rm{es}}}}} {\rm{Tr}}\left( {{{\bf{H}}_{\rm{e}}}{{\bf{W}}_{{\rm{bo}}}}} \right) = {\rm{Tr}}\left[ {\left( {{{\bf{H}}_{{\rm{es}}}} - {\rm{e}}_{\rm{h,UB}}{{\bf{I}}_{{M_1} \times {M_1}}}} \right){{\bf{W}}_{{\rm{bo}}}}} \right],
\end{equation}
where ${{\bf{I}}_{{M_{\rm{1}}}}}$ represents the identity matrix with rank being $M_1$.
\begin{proposition}\label{Proposition1}
The constraint on the security rate in (\ref{P1}d) can be converted to
\begin{equation}\label{SR_written}
\left( {\tau  - 1} \right){U_1}{U_2} \ge \tau {\rm{Tr}}\left[ {{{\bf{H}}_{{\rm{bo,min}}}}{{\bf{W}}_{\rm bo}}} \right]{\rm{Tr}}\left[ {{{\bf{H}}_{\rm{e}}}{{\bf{W}}_{\rm bo}}} \right],
\end{equation}
\begin{equation}\label{SR_1}
\begin{array}{l} \displaystyle
	{U_1} = {\rm{Tr}}\left( {{{\bf{H}}_{{\rm{bo,min}}}}{{\bf{W}}_{{\rm{bo}}}}} \right) + \left( {1 - \tau } \right){\sigma ^2}\\ \displaystyle \quad\;\,
	+ \left( {1 - \tau } \right){\rm{Tr}}\left( {{{\bf{G}}_{{\rm{bo,max}}}}{{\bf{O}}_{\rm{c}}} + {{\bf{G}}_{{\rm{bo,max}}}}\sum\nolimits_{n = 1}^N {{{\bf{O}}_n}} } \right),
\end{array}
\end{equation}
\begin{equation}\label{SR_2}
\begin{array}{l} \displaystyle
	{U_2} = \frac{\tau }{{\tau  - 1}}{\rm{Tr}}\left[ {{{\bf{H}}_{\rm{e}}}{{\bf{W}}_{{\rm{bo}}}}} \right] + \sigma _{\rm{e}}^2\\ \displaystyle \quad\;\,
	+ {\rm{Tr}}\left[ {{{\bf{G}}_{\rm{e}}}{{\bf{O}}_{\rm{c}}} + {{\bf{G}}_{\rm{e}}}\sum\nolimits_{n = 1}^N {{{\bf{O}}_n}} } \right] + {\rm{Tr}}\left( {{{\bf{H}}_{\rm{e}}}{{\bf{W}}_{{\rm{ta}}}}} \right),
\end{array}
\end{equation}
respectively, where $\tau  = {2^{{I_{\rm{S}}}}}$.
\end{proposition}
\begin{proof}
For brevity, the proof of Proposition 1 is provided to Appendix B. 
\end{proof}
Then, the right part of (\ref{SR_written}) can be rewritten as
	\begin{equation}\label{pro2}
	{\rm{Tr}}\left[ {{{\bf{H}}_{{\rm{bo,min}}}}{{\bf{W}}_{{\rm{bo}}}}} \right]{\rm{Tr}}\left[ {{{\bf{H}}_{\rm{e}}}{{\bf{W}}_{{\rm{bo}}}}} \right]{\rm{ = Tr}}{\left[ {{{\bf{h}}_{{\rm{bo}}}}{\bf{h}}_{\rm{e}}^{\rm{H}}{{\bf{W}}_{{\rm{bo}}}}} \right]^2}.
	\end{equation} 
\begin{proposition}\label{pro_re2}
Given above derivations, the following inequality holds: 
\begin{equation}\label{pro3-1}
	{U_1} + {U_2} \ge {\left\| {\begin{array}{*{20}{c}}
				{2\sqrt {{\tau  \mathord{\left/
								{\vphantom {\tau  {\left( {\tau  - 1} \right)}}} \right.
								\kern-\nulldelimiterspace} {\left( {\tau  - 1} \right)}}} {\rm{Tr}}\left( {{{\bf{h}}_{{\rm{bo}}}}{\bf{h}}_{\rm{e}}^{\rm{H}}{{\bf{W}}_{\rm bo}}} \right)}\\ \displaystyle
				{ - {U_1} + {U_2}}
		\end{array}} \right\|_2}.
\end{equation}
\end{proposition}
\begin{proof}
Please refer to Appendix C.  
\end{proof}

Based on the consideration of the imperfect CSI and (\ref{SR_2}), the range value of $U_2$ can be expressed as ${U_{2,{\rm{UB}}}} \ge {U_2} \ge {U_{2,{\rm{LB}}}}$, where
\begin{equation}\label{range_UB}
\begin{array}{l} \displaystyle
	{U_{2,{\rm{UB}}}} = \frac{\upsilon }{{\upsilon  - 1}}{\rm{Tr}}\left[ {\left( {{{\bf{H}}_{{\rm{es}}}} + {{\rm{e}}_{{\rm{h,UB}}}}{{\bf{I}}_{{M_1} \times {M_1}}}} \right){{\bf{W}}_{{\rm{bo}}}}} \right]\\ \displaystyle \quad\quad\;\;\,
	+ {\rm{Tr}}\left[ {{{\bf{T}}_{\rm{e}}}{{\bf{O}}_{\rm{c}}} + {{\bf{T}}_{\rm{e}}}\sum\nolimits_{n = 1}^N {{{\bf{O}}_n}} } \right] + {\sigma ^2}\\ \displaystyle \quad\quad\;\;\,
	+ {\rm{Tr}}\left[ {\left( {{{\bf{H}}_{{\rm{es}}}} + {{\rm{e}}_{{\rm{h,UB}}}}{{\bf{I}}_{{M_1} \times {M_1}}}} \right){{\bf{W}}_{{\rm{ta}}}}} \right],
\end{array}
\end{equation}
and
\begin{equation}\label{range_LB}
\begin{array}{l} \displaystyle
	{U_{2,{\rm{LB}}}} = \frac{\upsilon }{{\upsilon  - 1}}{\rm{Tr}}\left[ {\left( {{{\bf{H}}_{{\rm{es}}}} - {{\rm{e}}_{{\rm{h,UB}}}}{{\bf{I}}_{{M_1} \times {M_1}}}} \right){{\bf{W}}_{{\rm{bo}}}}} \right]\\ \displaystyle \quad\quad\;\;\,
	+ {\rm{Tr}}\left[ {{{\bf{T}}_{\rm{e}}}{{\bf{O}}_{\rm{c}}} + {{\bf{T}}_{\rm{e}}}\sum\nolimits_{n = 1}^N {{{\bf{O}}_n}} } \right] + {\sigma ^2}\\ \displaystyle \quad\quad\;\;\,
	+ {\rm{Tr}}\left[ {\left( {{{\bf{H}}_{{\rm{es}}}} - {{\rm{e}}_{{\rm{h,UB}}}}{{\bf{I}}_{{M_1} \times {M_1}}}} \right){{\bf{W}}_{{\rm{ta}}}}} \right].
\end{array}
\end{equation}
\begin{proposition}\label{pro_re_3}
Based on (\ref{pro3-1}), (\ref{range_UB}), and (\ref{range_LB}), the following inequality holds: 
\begin{equation}\label{pro4-1}
	{U_1} + {U_{{\rm{2,LB}}}} \ge {\left\| {\begin{array}{*{20}{c}}
				{2\sqrt {{\tau  \mathord{\left/
								{\vphantom {\tau  {\left( {\tau  - 1} \right)}}} \right.
								\kern-\nulldelimiterspace} {\left( {\tau  - 1} \right)}}} {\rm{Tr}}\left[ {{\bf{\tilde H}}{{\bf{W}}_{{\rm{bo}}}}} \right]}\\
				{ - {U_1} + {U_{{\rm{2,UB}}}}}
		\end{array}} \right\|_2},
\end{equation}
where ${\bf{\tilde H}} = {{\bf{h}}_{{\rm{bo}}}}{\bf{h}}_{{\rm{es}}}^{\rm{H}} + {{\rm{e}}_{{\rm{h,UB}}}}{{\bf{I}}_{{M_1} \times {M_1}}}$.
\end{proposition}
\begin{proof}
Please refer to Appendix D.  
\end{proof}

Let $\chi  = \left\{ {{x_1},{x_2},{x_3},{x_4},{x_5},{x_6},{x_7}} \right\}$ be a set of auxiliary variables further introduced, and the equivalent form of \textbf{P1} can be expressed as
\begin{subequations} \label{PB1}
	\begin{align}
		&{\rm{        }}{\bf{P4.1}}{\rm{.}}\;\mathop {\max }\limits_{{{\bf{o}}_{\rm{c}}},{{\bf{o}}_n},\chi } {x_1}\left( {{{\bf{o}}_{\rm{c}}},{{\bf{o}}_n}} \right)\\
		&{\rm{s}}{\rm{.t.}}\quad{x_2} - {x_3} - {x_4} \ge {x_1}{x_7},\\
		&\quad\quad\,\,  {\rm{Tr}}\left( {{{\bf{W}}_{{\rm{b1}}}}} \right) + {\rm{Tr}}\left( {{{\bf{O}}_{{\rm{c,}}n}}} \right) \le {x_7},\\
		& \quad\quad\,\, {\rm{Tr}}\left( {{{\bf{H}}_{{\rm{bo,min}}}}{{\bf{W}}_{{\rm{b1}}}}} \right) + {\rm{Tr}}\left[ {{{\bf{G}}_{{\rm{bo,max}}}}{{\bf{O}}_{{\rm{c,}}n}}} \right] \ge \left( {{2^{{x_2}}} - 1} \right){\sigma ^2},\\
		& \quad\quad\,\,{\rm{Tr}}\left[ {{{\bf{G}}_{{\rm{bo,max}}}}{{\bf{O}}_{{\rm{c,}}n}}} \right] \le \left( {{2^{{x_3}}} - 1} \right){\sigma ^2},\\
		&\quad\quad\,\, {\rm{Tr}}\left[ {\left( {{{\bf{G}}_{{\rm{es}}}} + {{\rm{e}}_{{\rm{g}},{\rm{UB}}}}{{\bf{I}}_{{M_2} \times {M_2}}}} \right){{\bf{O}}_{{\rm{c,}}n}}} \right] \le \left( {{2^{{x_6}}} - 1} \right){\sigma ^2},\\  
		&\quad\quad\,\, {\rm{Tr}}\left[ {{{\bf{G}}_{{\rm{es,min}}}}{{\bf{O}}_{{\rm{c,}}n}}} \right] + {\rm{Tr}}\left[ {{{\bf{H}}_{{\rm{es,min}}}}{{\bf{W}}_{{\rm{b1}}}}} \right] \ge \left( {{2^{{x_5}}} - 1} \right){\sigma ^2}, \\
	& \quad\quad\,\,   - \left( {{2^{{I_{\rm{c}}}}} - 1} \right)f\left( {{{\bf{O}}_n}} \right){\rm{ + Tr}}\left( {{{\bf{G}}_n}{{\bf{O}}_{\rm{c}}}} \right) \ge 0,\\
	&\quad\quad\,\,   - \left( {{2^{{I_{\rm{p}}}}} - 1} \right)f\left( {{{\bf{O}}_j}} \right){\rm{ + Tr}}\left( {\sum\nolimits_{n = 1}^N {{{\bf{G}}_n}{{\bf{O}}_n}} } \right) \ge 0,\\
	&\quad\quad\,\,  {x_5} - {x_6} \le {x_4},\\
	&\quad\quad\,\, {\rm rank}\left( {{{\bf{W}}_{{\rm{b1}}}}} \right) = {\rm rank}\left( {{{\bf{O}}_{\rm{c}}}} \right) ={\rm rank}\left( {{{\bf{O}}_n}} \right) = 1,\\
	&\quad\quad\,\, (\ref{P1} \rm b), (\ref{P1} \rm c), (\ref{pro4-1}), 
	\end{align}	
\end{subequations}
where $f\left( {{{\bf{O}}_n}} \right) = \left[ {{\rm{Tr}}\left( {\sum\nolimits_{n = 1}^N {{{\bf{G}}_n}{{\bf{O}}_n}} } \right) + {\rm{Tr}}\left( {{{\bf{H}}_n}{{\bf{W}}_{{\rm{b1}}}}} \right) + \sigma _n^2} \right]$ and ${{\bf{O}}_{{\rm{c,}}n}} = {{\bf{O}}_{\rm{c}}} + \sum\nolimits_{n = 1}^N {{{\bf{O}}_n}}$.
Notice that (\ref{PB1}\rm b) is non-convex; thus another variable $x_8$ is provided to convert (\ref{PB1}\rm b) into ${x_2} - {x_3} - {x_4} \ge {x_8 ^2}$ and $f\left( {{x_7},{x_8}} \right) = x_1^{ - 1}x_7^{ - 1}x_8^2 \ge 1$, which leads to 
\begin{equation}\label{C1_transformation}
		\frac{{{x_2} - {x_3} - {x_4} + 1}}{2} \ge {\left\| {\begin{array}{*{20}{c}} \displaystyle
					{\frac{{{x_2} - {x_3} - {x_4} - 1}}{2}}\\ \displaystyle
					x_8 
			\end{array}} \right\|_2}.
\end{equation}

With the fixed $x_7$ and $x_8$, $x_{7,0}$ and $x_{8,0}$ taken into account, we hold that
\begin{equation}\label{specific_value}
		2{x_{7,0}}{x_{8,0}}{x_8} \ge x_{7,0}^2{x_1} + x_{8,0}^2{x_7}.
\end{equation}
(\ref{specific_value}) enables (\ref{PB1}\rm e) and (\ref{PB1}\rm f) to be respectively expressed as convex forms as follows
\begin{equation}\label{PB1_e_convex}
\begin{array}{l} \displaystyle
	{\rm{Tr}}\left[ {{{\bf{G}}_{{\rm{bo,max}}}}\left( {{{\bf{O}}_{\rm{c}}} + {{\sum\nolimits_{n = 1}^N {\bf{O}} }_n}} \right)} \right] \le \\ \displaystyle
	{\sigma ^2}\left[ {{2^{{x_{3,0}}}}\left( {{x_3}\ln 2 - {x_{3,0}}\ln 2 + 1} \right) - 1} \right],
\end{array}
\end{equation}
\begin{equation}\label{PB1_f_convex}
\begin{array}{l} \displaystyle
	{\rm{Tr}}\left[ {\left( {{{\bf{G}}_{{\rm{es}}}} + {{\rm{e}}_{{\rm{g,UB}}}}{{\bf{I}}_{{M_2} \times {M_2}}}} \right)\left( {{{\bf{O}}_{\rm{c}}} + \sum\nolimits_{n = 1}^N {{{\bf{O}}_n}} } \right)} \right]\\ \displaystyle
	\le {\sigma ^2}\left[ {{2^{{x_{6,0}}}}\left( {{x_6}\ln 2 - {x_{6,0}}\ln 2 + 1} \right) - 1} \right].
\end{array}
\end{equation}

Let us define the $i$-th iteration of ${{\bf{o}}_{\rm{c}}}$ and ${{\bf{o}}_n}$ are defined as ${{\bf{o}}_{{\rm{c}},i}}$ and ${{\bf{o}}_{n,i}}$. Then, the corresponding approximate matrices can be obtained as ${{\bf{O}}_{{\rm{c,}}i}} = {{\bf{o}}_{{\rm{c,}}i}}{\bf{o}}_{\rm{c}}^{\rm{H}} + {{\bf{o}}_{\rm{c}}}{\bf{o}}_{{\rm{c,}}i}^{\rm{H}} - {{\bf{o}}_{{\rm{c,}}i}}{\bf{o}}_{{\rm{c,}}i}^{\rm{H}}$ and ${{\bf{O}}_{n,i}} = {{\bf{o}}_{n,i}}{\bf{o}}_n^{\rm{H}} + {{\bf{o}}_n}{\bf{o}}_{n,i}^{\rm{H}} - {{\bf{o}}_{n,i}}{\bf{o}}_{n,i}^{\rm{H}}$.
Finally, the non-convex targeted optimization \textbf{P1} can be rewritten in a convex-form as
\begin{subequations} \label{PB2}
	\begin{align}
		&{\rm{        }}{\bf{P4}}.\;\mathop {\max }\limits_{{{\bf{o}}_{\rm{c}}},{{\bf{o}}_n},\chi } {x_1}\\
		&{\rm{s}}{\rm{.t.}}\quad (\ref{P1} \rm b), (\ref{P1} \rm c), (\ref{pro4-1}),  (\ref{PB1}\rm c), (\ref{PB1}\rm d), (\ref{PB1}\rm g)~(\ref{PB1}\rm k),\\
		&\quad\quad\,\, (\ref{specific_value}), (\ref{PB1_f_convex}).
	\end{align}	
\end{subequations}
The calculation process of (\ref{PB2}) is presented in Algorithm \ref{algorithm_3}. 
Finally, the complete algorithm for
solving (\ref{P1}) is outlined in Algorithm \ref{algorithm_4}, where
Algorithms \ref{algorithm_1}-\ref{algorithm_3} are alternately executed in each iteration.

\begin{algorithm}[t]
	\caption{Optimization of the transmit BF vectors at $\rm BS_2$.}
	\label{algorithm_3}
	
	{\algorithmicrequire} Given $\textbf{a}$, $\textbf{w}_{\rm bo}$, and $\textbf{w}_{\rm ta}$, and related variables in the defined models.
	
	{\algorithmicensure} Optimal transmit signal BF vectors $\textbf{o}_{\rm c}$, and $\textbf{o}_n$.
	
	\begin{algorithmic}[1] 
		\State Set a set of auxiliary variables $\chi$ and $x_8$, iteration threshold $\delta$, $\varepsilon=0$, ${\rm SEE}_{0}=0$, and ${\rm SEE}_{-1}=0$;
		\State Set positive values ${\rm e _h}$, ${\rm e _g}$, ${\rm e _{\rm h, UB}}$, ${\rm e _{\rm g, UB}}$, ${I _{\rm{S}}}$, ${I _{\rm{c}}}$, ${I _{\rm{p}}}$, $P_{\rm 1, max}$, and $P_{\rm 2, max}$;
		\State Initialize  $\textbf{O}_{\rm c, \varepsilon}$ and $\textbf{O}_{n,\varepsilon}$ to satisfy constraints in (\ref{PB2}\rm b) and (\ref{PB2}\rm c);
		\While {all constraints in (\ref{PB2}\rm b) and (\ref{PB2}\rm c) are satisfied}
		\State Set $i = 0$, and initialize ${x_{2,0,i }}$, ${x_{6,0,i }}$, ${x_{7,0,i }}$, and ${x_{8,0,i }}$;
		\Repeat 
		\State $i  = i  + 1$;
		\State Solve the optimization objective in (\ref{PB2}) to obtain 
		\Statex \qquad \quad maximum SEE;
		\State Update $\textbf{O}_{\rm c, {\varepsilon}}$, $\textbf{O}_{n,\varepsilon}$, ${x_{2,0,i }}$, ${x_{6,0,i }}$, ${x_{7,0,i}}$, and ${x_{8,0,i }}$;
		\Until {${\rm{SE}}{{\rm{E}}_\varepsilon } - {\rm{SE}}{{\rm{E}}_{\varepsilon  - 1}} \le {\delta }$};
		\State Obtain solutions $\textbf{O}_{\rm c}$ and $\textbf{O}_{n}$;
		\State Set ${{\bf{O}}_{{\rm{c}},\varepsilon  + 1}}{\rm{ = }}{{\bf{O}}_{\rm{c}}}$ and ${{\bf{O}}_{{{n}},\varepsilon  + 1}}{\rm{ = }}{{\bf{O}}_{{n}}}$;
		
		\While {${{\bf{O}}_{{\rm{c}},\varepsilon  + 1}} \approx {{\bf{O}}_{{\rm{c}},\varepsilon }}$, and ${{\bf{O}}_{{{n}},\varepsilon  + 1}} \approx {{\bf{O}}_{{{n}},\varepsilon }}$  are not \indent \indent  satisfied}
		\State $\varepsilon  = \varepsilon  + 1$;
		\EndWhile
		
		\EndWhile
		\State Employ SVD to $\textbf{O}_{\rm c, \varepsilon}$ and $\textbf{O}_{n,\varepsilon}$, and then the transmit signal BF vectors $\textbf{o}_{\rm c}$ and $\textbf{o}_n$ are ascertained, respectively;
	\end{algorithmic}
\end{algorithm} 

\begin{algorithm}[t]
	\caption{Optimization of the SEE and the weighted sum of $\gamma_{\rm b1}$ and $R_{\rm S}$ in (\ref{P1}) based on Algorithms 1, 2, and 3.}
	\label{algorithm_4}
	
	{\algorithmicrequire} Initialization values of $\textbf{a}$, $\textbf{w}_{\rm bo}$, $\textbf{w}_{\rm ta}$, $\textbf{o}_{\rm c}$, and $\textbf{o}_{n}$, and related variables in the defined models, the number of iterations $\varepsilon $, and tolerance threshold $\delta$, respectively.
	
	{\algorithmicensure} BF vectors $\textbf{a}$, $\textbf{w}_{\rm bo}$, $\textbf{w}_{\rm ta}$, $\textbf{o}_{\rm c}$, and $\textbf{o}_{n}$, respectively.
	
	\begin{algorithmic}[1] 
		\Repeat 
		\State Carry out Algorithm \ref{algorithm_1} to obtain $\textbf{a}_{ \varepsilon+1 }$ with $\textbf{w}_{\rm bo, \varepsilon }$, \Statex\qquad$\textbf{w}_{\rm ta, \varepsilon}$, $\textbf{o}_{\rm c, \varepsilon}$, and $\textbf{o}_{n, \varepsilon}$;
		\State Carry out Algorithm \ref{algorithm_2} to obtain $\textbf{w}_{\rm bo, \varepsilon+1}$ and $\textbf{w}_{\rm ta, \varepsilon+1}$ 
		\Statex\qquad with $\textbf{a}_{ \varepsilon+1 }$, $\textbf{o}_{\rm c, \varepsilon}$, and $\textbf{o}_{n, \varepsilon}$;
		\State Carry out Algorithm \ref{algorithm_3} to obtain $\textbf{o}_{\rm c, \varepsilon+1}$ and $\textbf{o}_{n, \varepsilon+1}$ 
		\Statex\qquad with $\textbf{a}_{ \varepsilon+1 }$, $\textbf{w}_{\rm bo, \varepsilon+1}$, and $\textbf{w}_{\rm ta, \varepsilon+1}$;
		\State $\varepsilon  = \varepsilon  + 1$;
		\Until Variation of adjacent iteration results is less than $\delta$;
		\State Employ QR decomposition to obtain echo signal BF vector $\textbf{a}$, employ SVD to $\textbf{W}_{\rm bo, \varepsilon}$, $\textbf{W}_{\rm ta,\varepsilon}$, $\textbf{O}_{\rm c, \varepsilon}$, and $\textbf{O}_{n,\varepsilon}$, and then the transmit signal BF vectors $\textbf{w}_{\rm bo}$, $\textbf{w}_{\rm ta}$, $\textbf{o}_{\rm c}$, and $\textbf{o}_n$ are ascertained, respectively;
	\end{algorithmic}
\end{algorithm} 

\subsection{Complexity analysis of Algorithm \ref{algorithm_4}}
For Algorithm \ref{algorithm_4}, there exist $M_1^2 + M_2^2\left( {N + 1} \right) + 1$ original variables, $\left( {N + 6} \right)$ slack variables, 2 $M_1$-size \emph{linear matrix inequality} (LMI) constraints, $(N+1)$ $M_2$-size LMI constraints, $(N+8)$ 1-size LMI constraints, and $(N+1)$ 3-size SOC constraints. Therefore, the computational complexity of Algorithm \ref{algorithm_4} is given in (\ref{Computational_complexity}) at the top of the next page, where $D  = {M_1^2 + \left( {N + 1} \right)M_2^2 + \left( {N +7} \right)}$.

\begin{figure*}[htbp]
	\begin{equation}\label{Computational_complexity}
{\cal O}\left\{ {\sqrt {2{M_1} + \left( {N + 1} \right){M_2} + \left( {N + 8} \right)} D\left[ {M_1^2\left( {{M_1} + D} \right) + \left( {N + 1} \right)M_2^2  \left( {{M_2} + D} \right) + \left( {N + 8} \right)\left( {1 + D} \right) + 3\left( {N + 1} \right) + {D^2}} \right]} \right\},
	\end{equation}
	\hrulefill
\end{figure*}

\subsection{Convergence and Scalability of Algorithm \ref{algorithm_4}}
The proposed alternating optimization algorithm is guaranteed to converge.  This is because, the SEE is upper-bounded due to the finite channel gains, bandwidth, transmit power, and total power consumption. Furthermore, in each iteration, Algorithms 1-3 optimally solve their respective subproblems, leading to a non-decreasing sequence of the objective value SEE. Hence, Algorithm 4 that is non-decreasing and bounded above can be guaranteed to converge to a finite limit, i.e. a feasible solution \cite{optimization}.

The proposed algorithm not only achieves a notable balance between complexity and reliability, but also has the scalability. On the one hand, for a better algorithmic acceleration, we can reduce iterations by 35-40\% using previous solutions, stop optimization when SEE improvement $\le$ 0.1\% per iteration, and exploit inherent parallelism in Algorithms 1-3 to fulfill parallel subproblem solution. On the other hand, we can use GPU acceleration since the matrix operations in Algorithms 1-3 are highly amenable to parallel processing. The alternating optimization structure naturally supports the distributed implementation fixed-point arithmetic is adopted owing to its potential for significant speedup with the acceptable precision loss.


\begin{table*}[htbp]
	\centering
	\caption{\label{tab:parameter}Main simulation parameters \cite{ISAC_PLS_14,ISAC_PLS_13}.}
	\begin{tabular}{cc||cc}
		\hline\hline 
		\textbf{Parameter} & \textbf{Value} & \textbf{Parameter} & \textbf{Value} \\
		\hline\hline
		Baseband signal bandwidth (${f_{\rm m}}$)	& 80 MHz  & Carrier frequency ($f_{\rm c}$)	& 18 GHz \\
		\hline
		Number of transmit antennas at $\rm BS_1$ ($M_1$) 	& 12 &Number of transmit antennas at $\rm BS_2$ ($M_2$) 	& 12\\
		\hline
		Coverage region of a cell ($r$) 	& 120 m & Number of multicast communication users ($N$) & 4 \\
		\hline
		Number of ISAC sensing target & 1 & Number of ISAC communication user  & 1  \\
		\hline
		Sensing SCNR threshold ($\gamma_{\rm th}$)	& 0.5 & Rate thresholds ($I_{\rm S}$, $I_{\rm c}$, $I_{\rm p}$) & 1 bps/Hz  \\
		\hline
		Maximum transmit power at BSs ($P_{\rm 1,max}$, $P_{\rm 2,max}$) & 25 dBW  & Noise power at each node ($\sigma^2$) & -80 dBW   \\
		\hline\hline 
	\end{tabular}
\end{table*}
\section{Numerical Results and Discussions}
In this section, simulations are performed to illustrate the effectiveness of the presented heterogeneous ISAC and the BF optimization scheme. The key simulation parameters are provided in Table \ref{tab:parameter}, respectively. Note that we compare the proposed RSMA-based scheme against the \emph{power-domain NOMA} (PD-NOMA), which employs the SIC at receivers and optimizes the power allocation among users. While the NOMA literature includes variants with explicit multi-casting capabilities, the benchmark focuses on the widely-adopted uni-casting oriented PD-NOMA to maintain a clear and interpretable comparison framework.


\begin{figure*}[htbp]
	\centering
	\subfigure[SEE vs. iteration number vs. $M_1$ for different convergence schemes.]{
		\label{iteration_com}
		\begin{minipage}[t]{0.25\linewidth}
			\centering
			\includegraphics[width=2in]{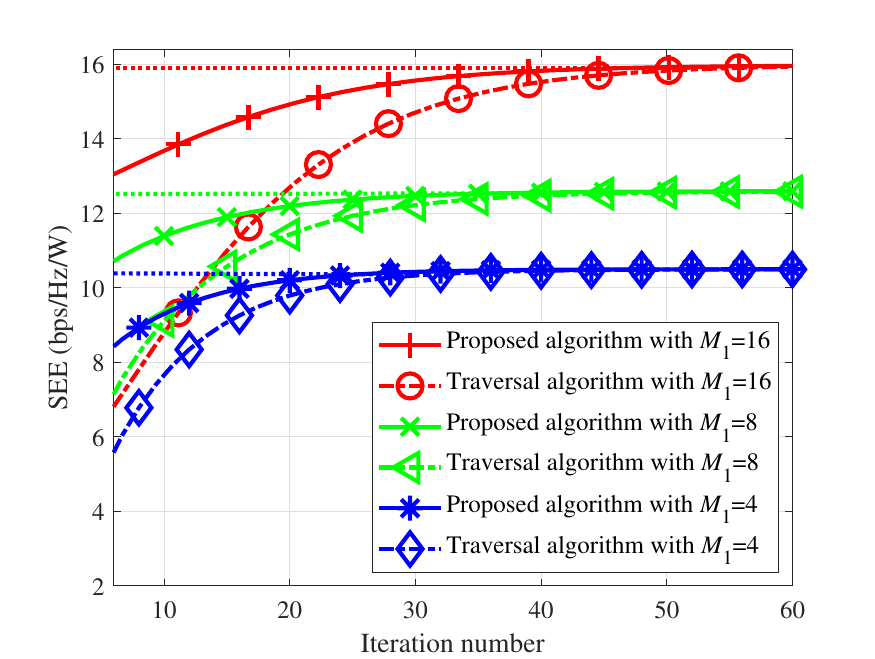}
		\end{minipage}%
	}%
	\subfigure[$\gamma_{\rm b1}$ vs. $P_{\rm 1}$ for different multiple access schemes.]{
		\label{SCNR_com}
		\begin{minipage}[t]{0.25\linewidth}
			\centering
			\includegraphics[width=2in]{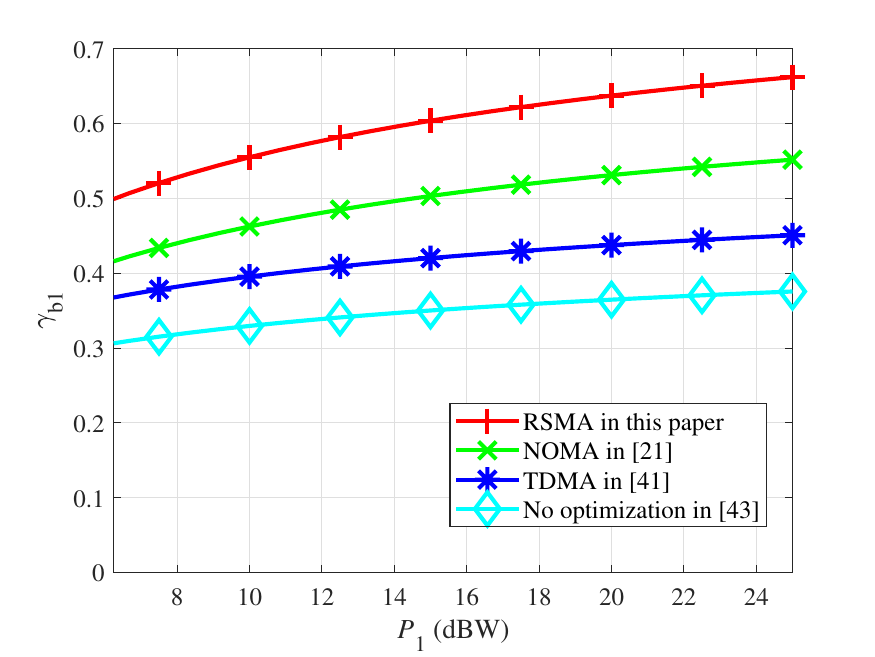}
		\end{minipage}%
	}%
	\subfigure[SEE vs. $P_{\rm 2}$ for different CSI uncertainty parameters.]{
		\label{SEE_ICSI_com}
		\begin{minipage}[t]{0.25\linewidth}
			\centering
			\includegraphics[width=2in]{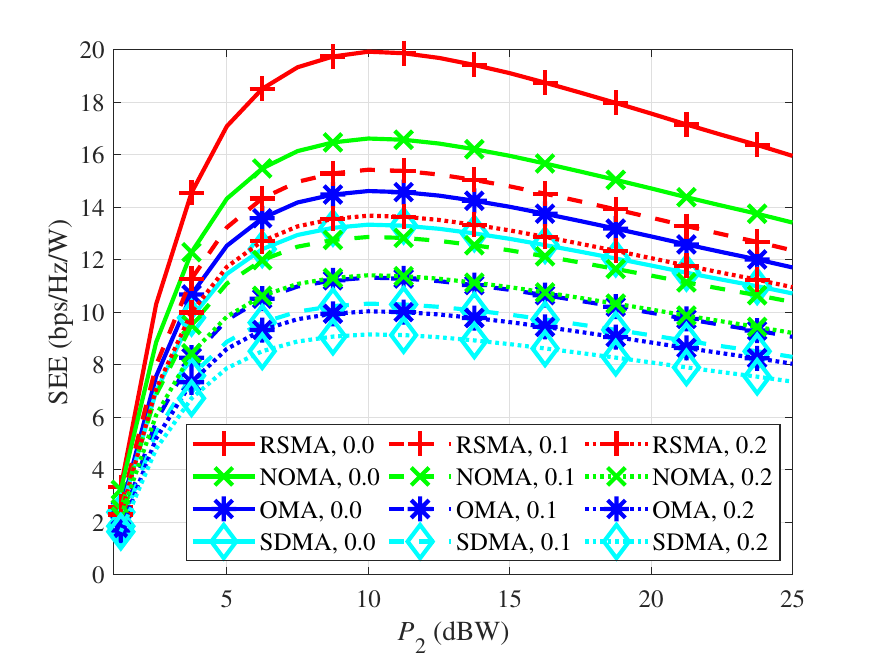}
		\end{minipage}
	}%
    \subfigure[SCNR-security-rate Pareto front.]{
    	\label{SCNR_RS}
    	\begin{minipage}[t]{0.25\linewidth}
    		\centering
    		\includegraphics[width=5cm]{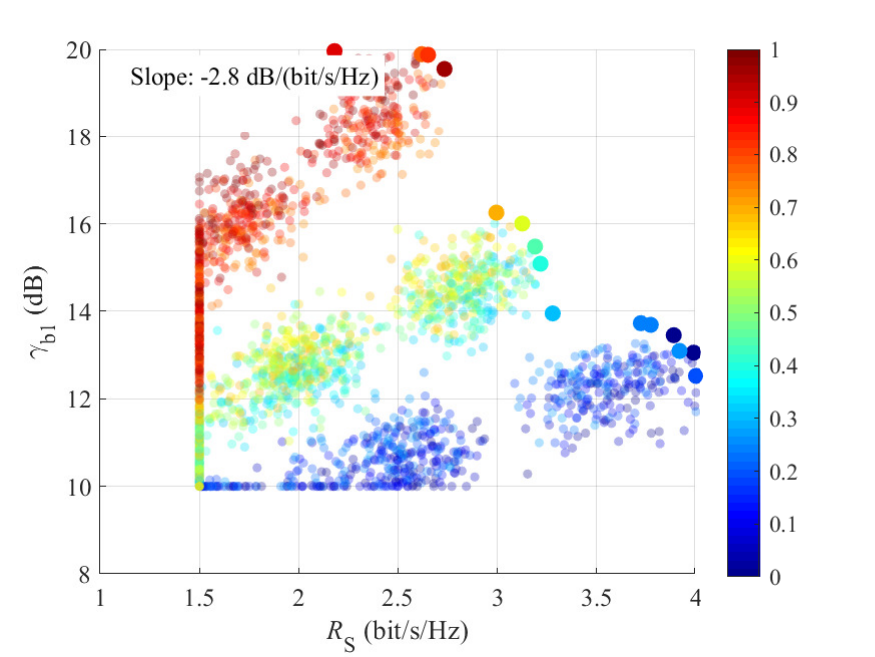}
    	\end{minipage}%
    }%
	\centering
	\caption{Illustration of the basic performance.}
	\label{fourfig_1}
\end{figure*}

\begin{figure*}[htbp]
	\centering
	\subfigure[SEE vs. $P_{\rm 1}$ for different transmit BF optimization schemes at $\rm BS_1$.]{
		\label{SEE_opti_com}
		\begin{minipage}[t]{0.25\linewidth}
			\centering
			\includegraphics[width=2in]{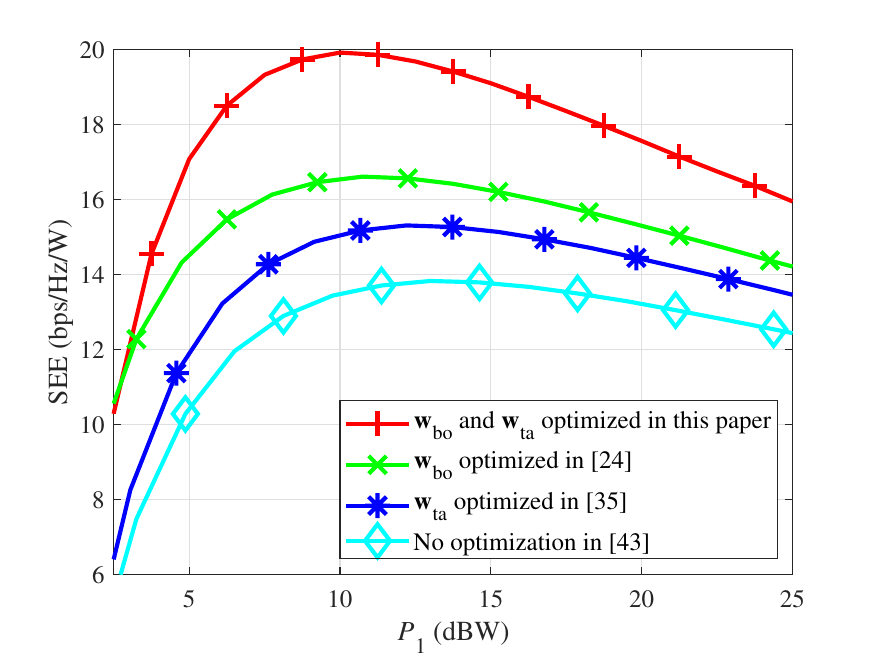}
		\end{minipage}%
	}%
\subfigure[SEE vs. $P_2$ vs. ($M_1$,$M_2$) with RSMA.]{
	\label{SEE_antenna}
	\begin{minipage}[t]{0.25\linewidth}
		\centering
		\includegraphics[width=2in]{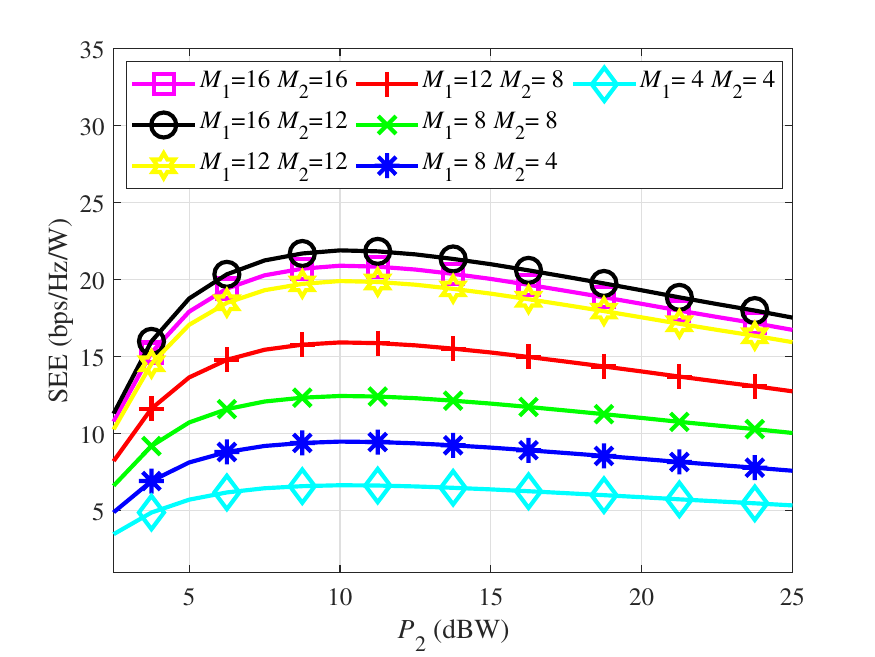}
	\end{minipage}%
}%
	\subfigure[SEE vs. $P_2$ vs. $N$ with RSMA.]{
		\label{SEE_user}
		\begin{minipage}[t]{0.25\linewidth}
			\centering
			\includegraphics[width=2in]{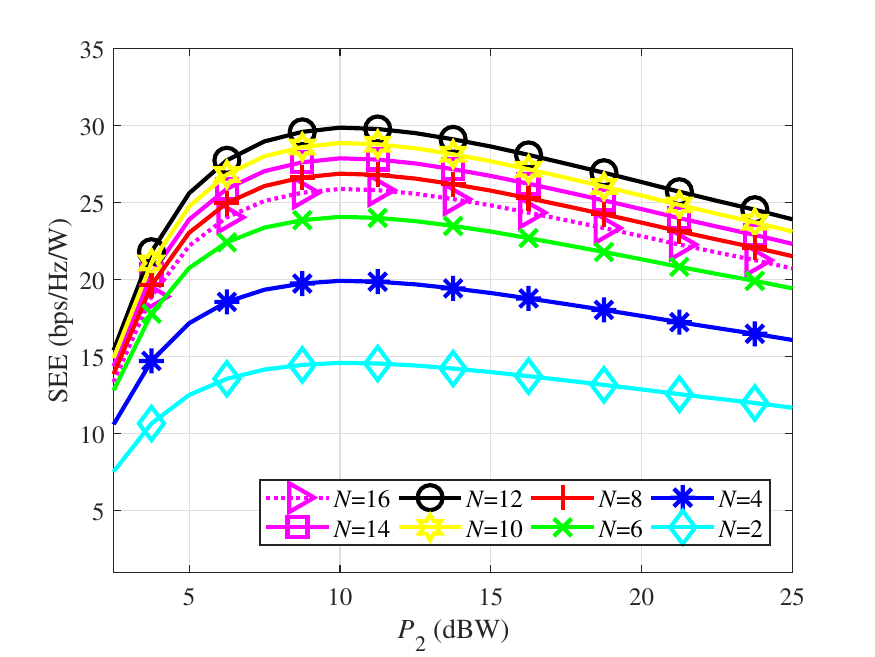}
		\end{minipage}%
	}%
	\subfigure[SEE vs. $P_2$ vs. ($I_{\rm S}$,$I_{\rm c}$,$I_{\rm p}$) with RSMA.]{
		\label{SEE_thresholds}
		\begin{minipage}[t]{0.25\linewidth}
			\centering
			\includegraphics[width=2in]{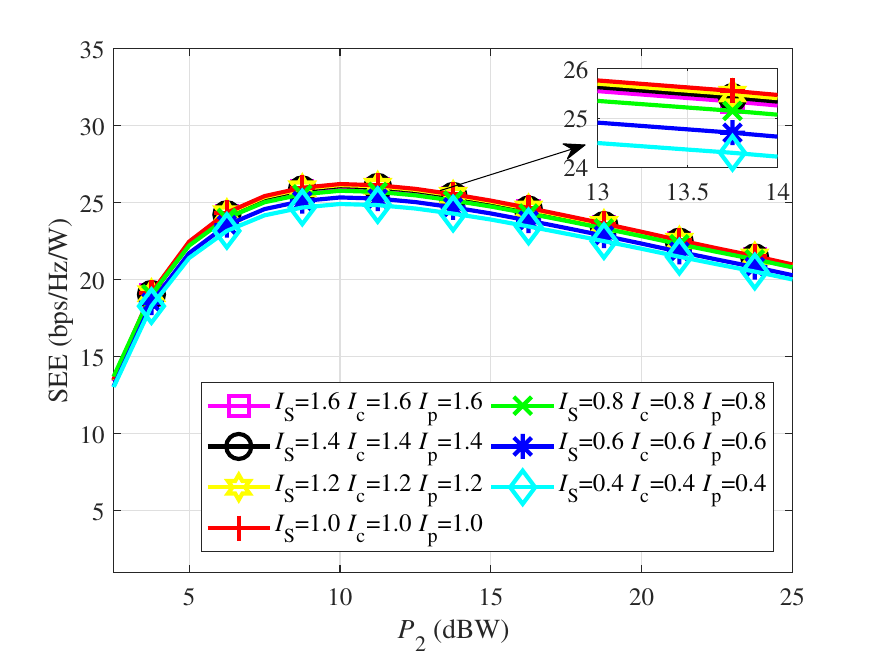}
		\end{minipage}
	}%
	\centering
	\caption{Sensitive parameter validation.}
	\label{sensitive_parameter}
\end{figure*}

\begin{figure*}[htbp]
	\centering
	\subfigure[Processing latency under different schemes.]{
		\label{SEE_latency}
		\begin{minipage}[t]{0.25\linewidth}
			\centering
			\includegraphics[width=2in]{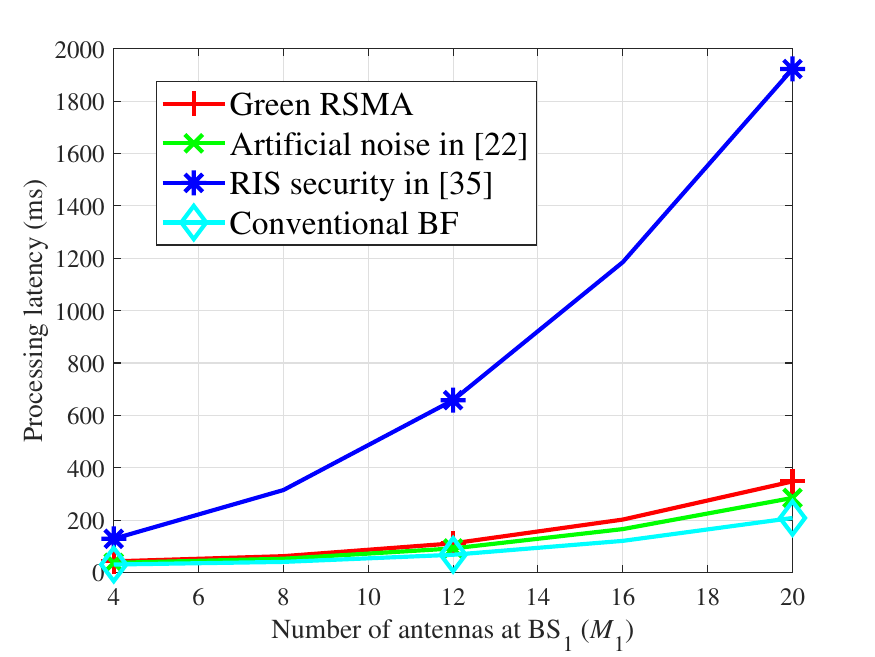}
		\end{minipage}%
	}%
	\subfigure[Hardware implementation requirement under different schemes.]{
		\label{SEE_memory}
		\begin{minipage}[t]{0.254\linewidth}
			\centering
			\includegraphics[width=2in]{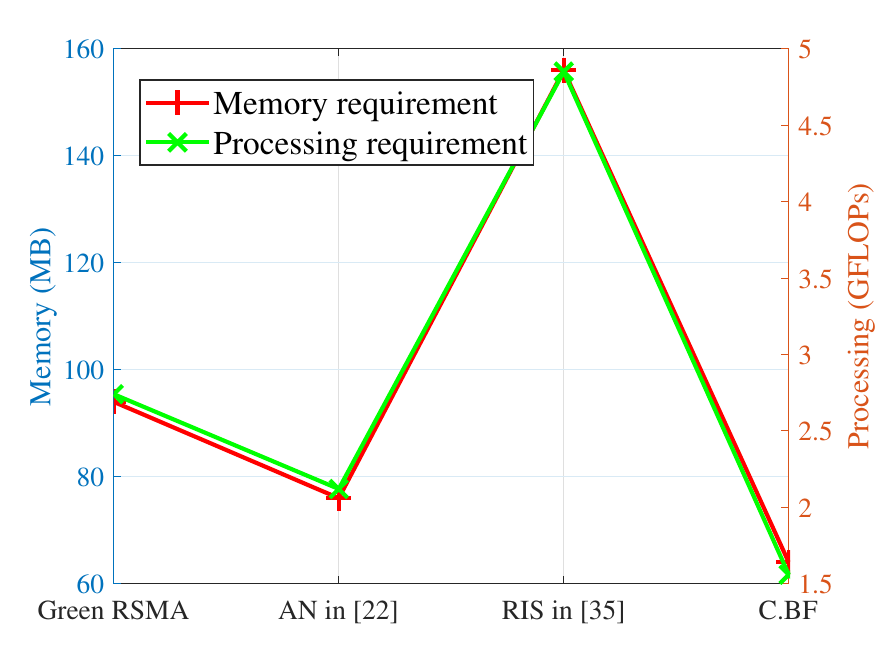}
		\end{minipage}%
	}%
	\subfigure[Cross-cell SINR performance under different schemes.]{
		\label{SEE_cross_interference}
		\begin{minipage}[t]{0.246\linewidth}
			\centering
			\includegraphics[width=2in]{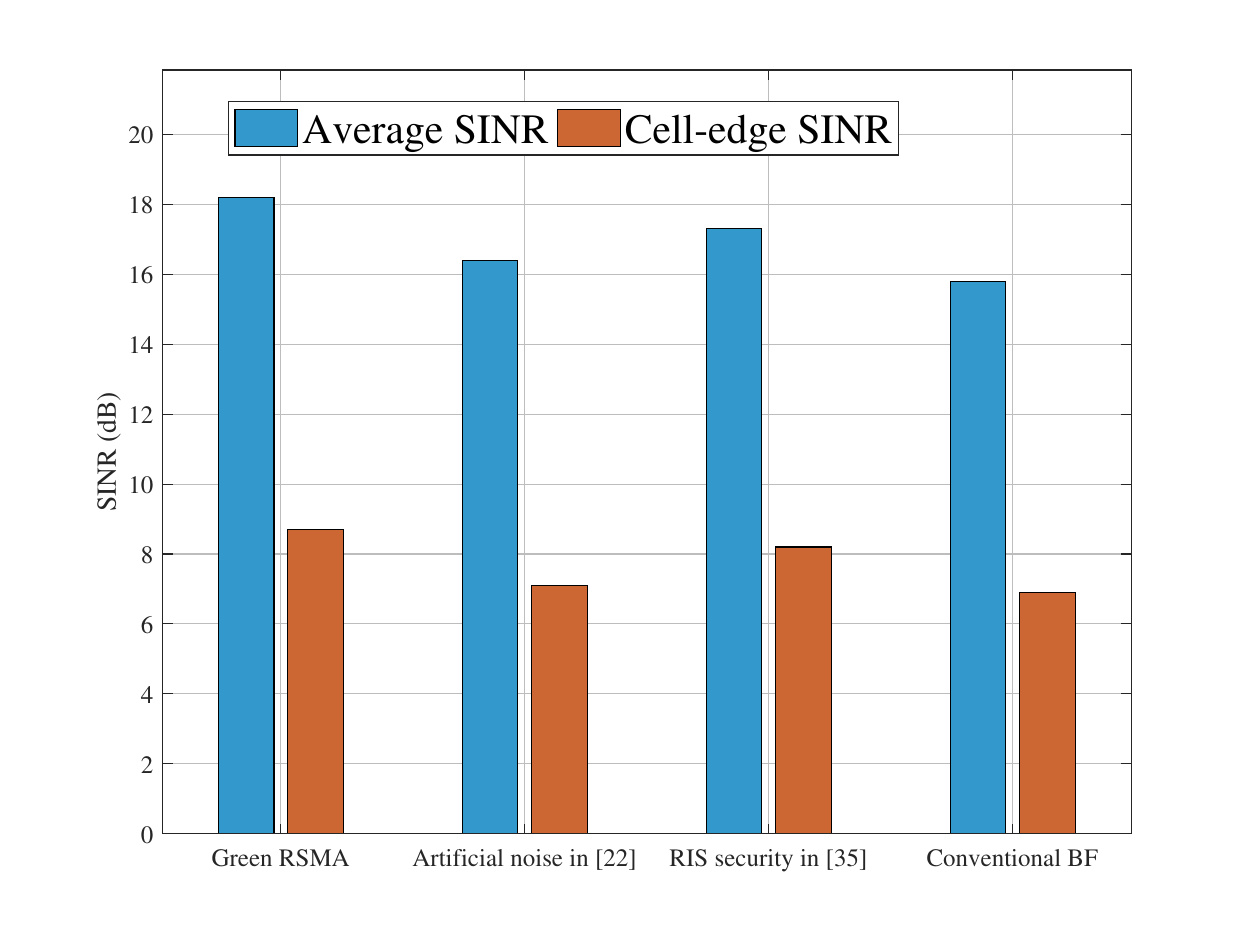}
		\end{minipage}
	}%
	\subfigure[Runtime vs. number of antennas under different schemes.]{
		\label{SEE_runtime_antenna}
		\begin{minipage}[t]{0.25\linewidth}
			\centering
			\includegraphics[width=2in]{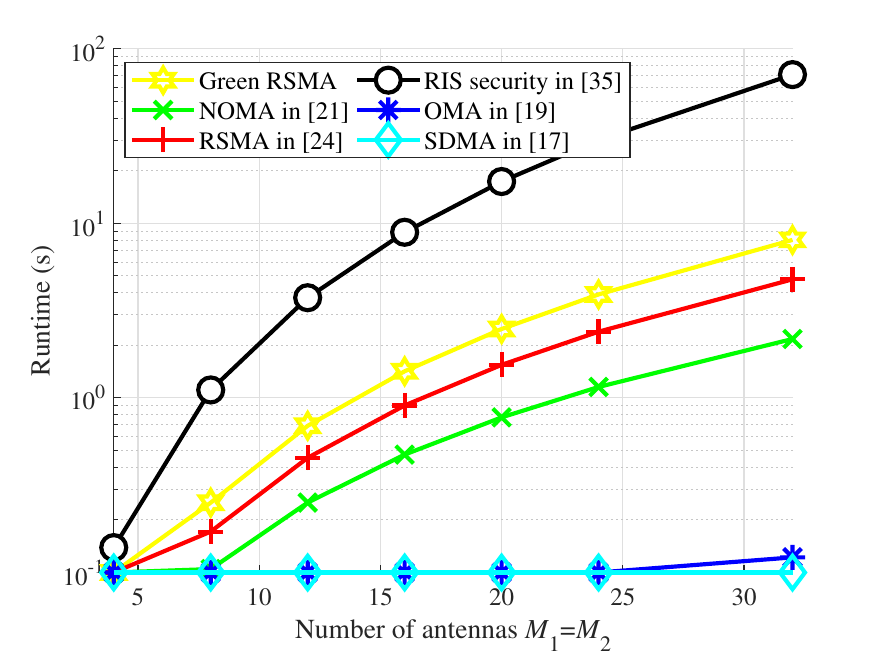}
		\end{minipage}
	}%
	\centering
	\caption{Effectiveness and scalability validation.}
	\label{algorithm_cost}
\end{figure*}

\begin{figure*}[htbp]
	\centering
	\subfigure[Centralized deployment.]{
		\label{central_de}
		\begin{minipage}[t]{0.25\linewidth}
			\centering
			\includegraphics[width=2in]{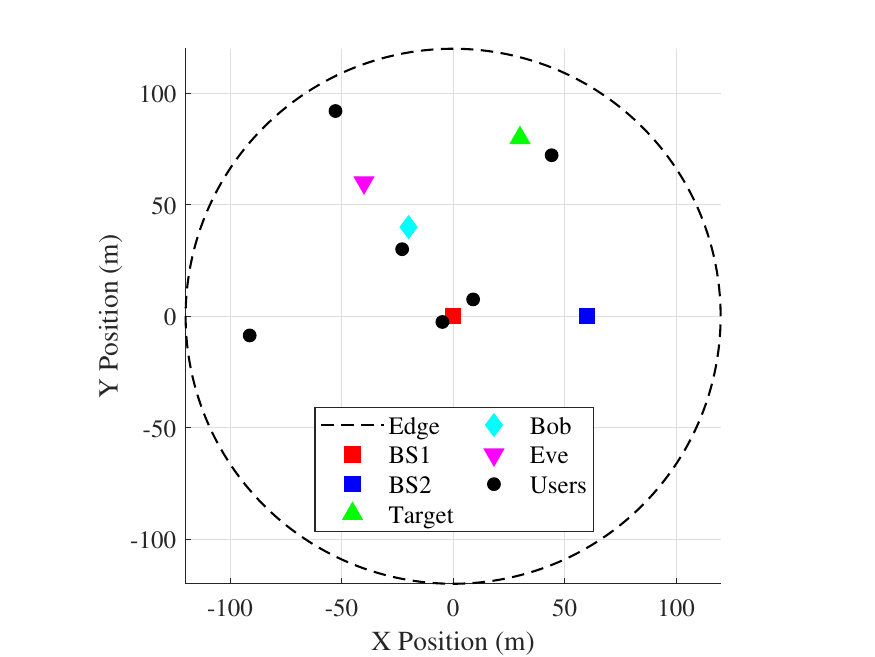}
		\end{minipage}%
	}%
	\subfigure[Edge-clustered deployment.]{
		\label{edge_de}
		\begin{minipage}[t]{0.25\linewidth}
			\centering
			\includegraphics[width=2in]{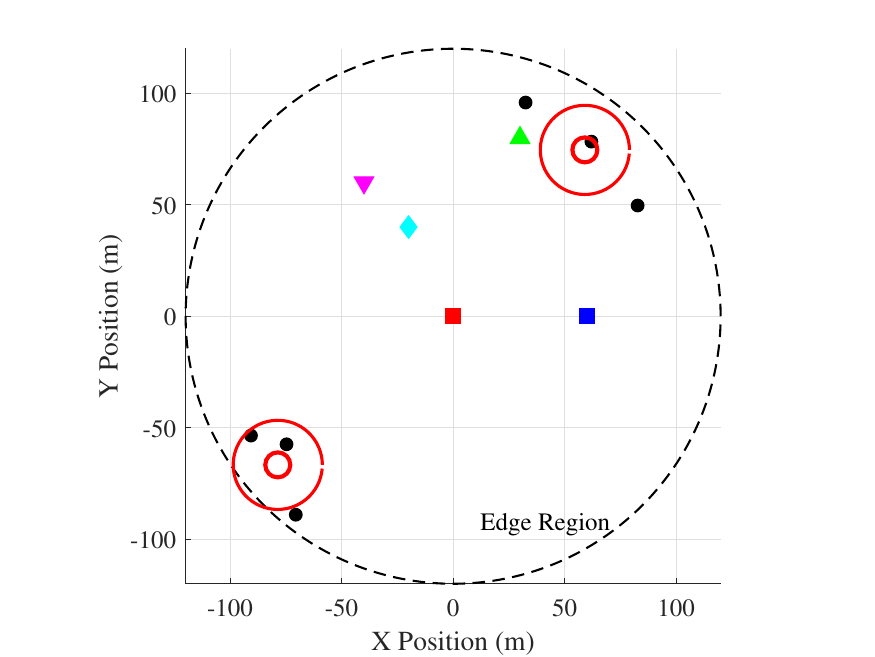}
		\end{minipage}%
	}%
	\subfigure[Hybrid deployment.]{
		\label{hybrid_dp}
		\begin{minipage}[t]{0.25\linewidth}
			\centering
			\includegraphics[width=2in]{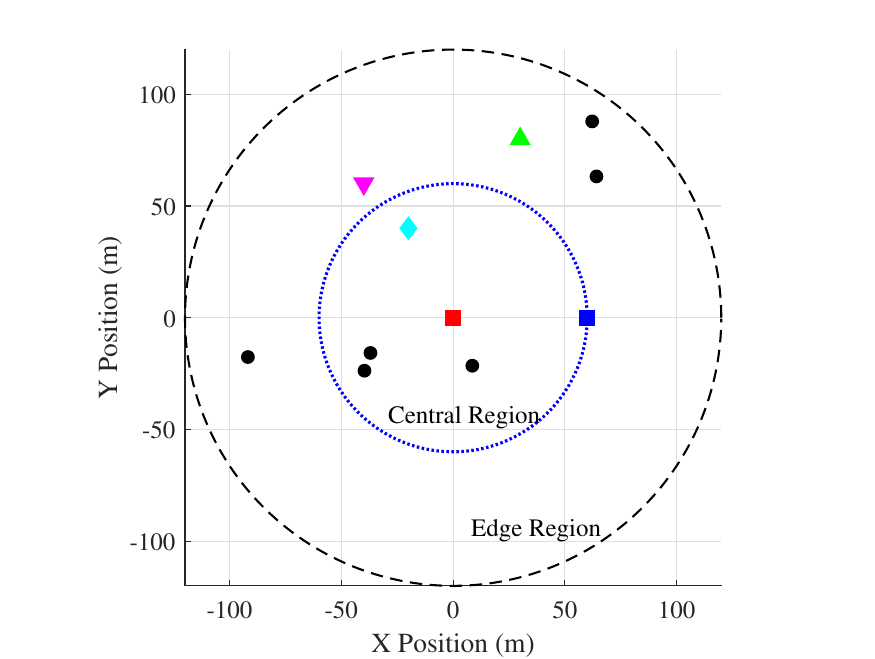}
		\end{minipage}
	}%
	\subfigure[SEE vs. multiple access vs. deployment scenario.]{
		\label{SEE_topology}
		\begin{minipage}[t]{0.25\linewidth}
			\centering
			\includegraphics[width=2in]{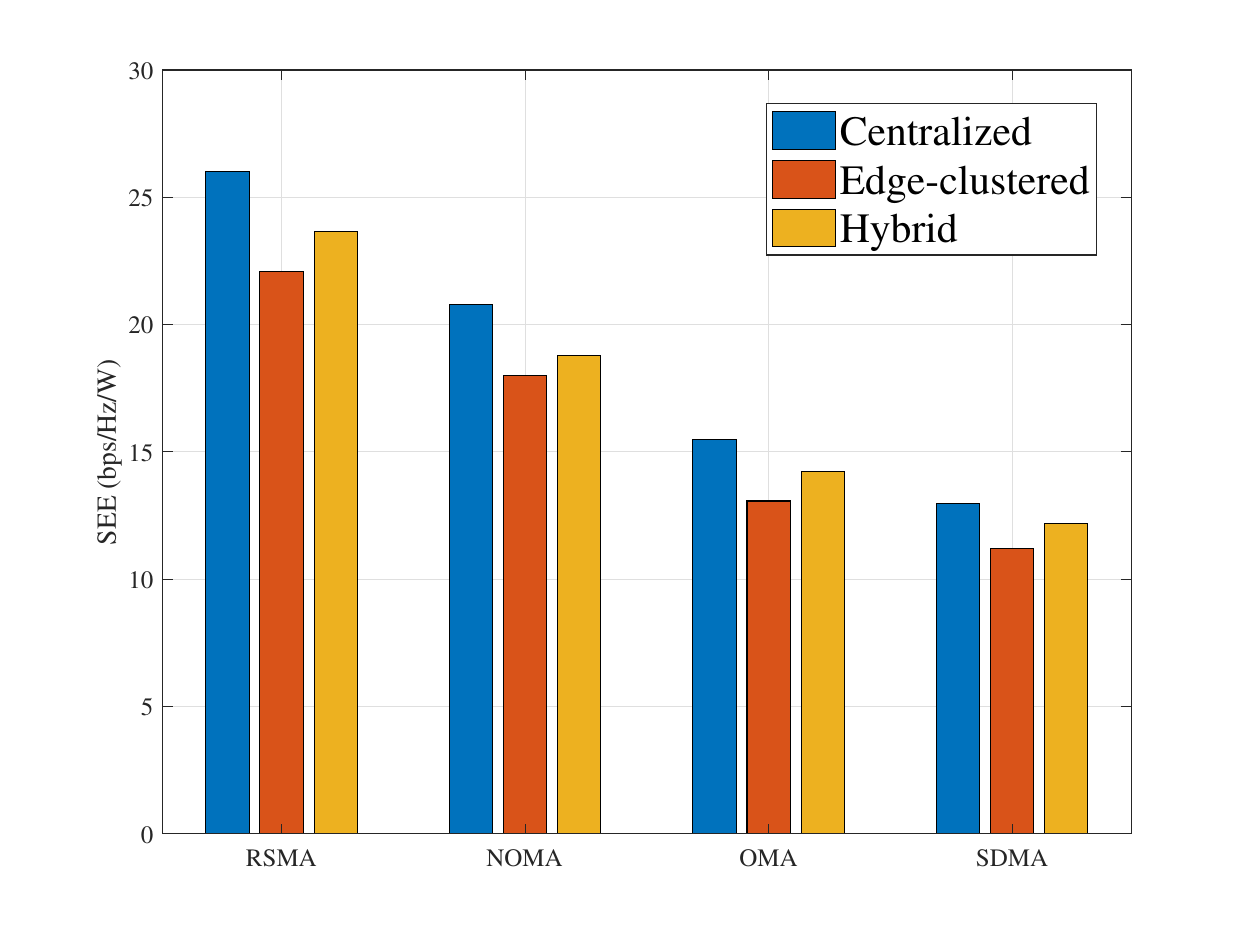}
		\end{minipage}
	}%
	\centering
	\caption{Robustness validation for varying network topologies.}
	\label{topo_com}
\end{figure*}

\begin{figure*}[htbp]
	\centering
	\subfigure[SEE vs. $P_{\rm 2}$ for different channel fading models.]{
		\label{channel_model_com}
		\begin{minipage}[t]{0.25\linewidth}
			\centering
			\includegraphics[width=5cm]{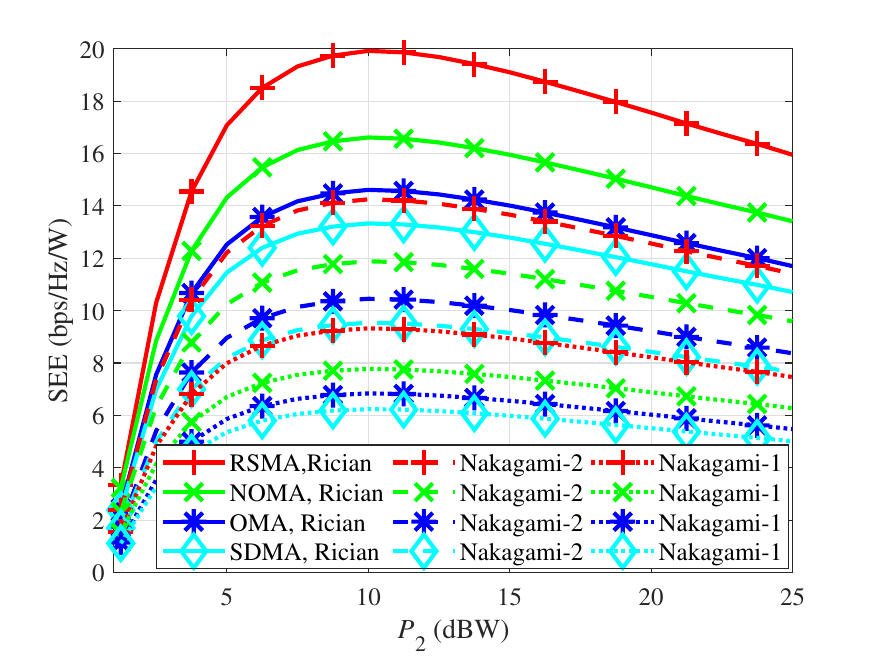}
		\end{minipage}%
	}%
\subfigure[SEE vs. $P_{\rm 2}$ for different uncertainty models.]{
	\label{uncertainty_model}
	\begin{minipage}[t]{0.25\linewidth}
		\centering
		\includegraphics[width=5cm]{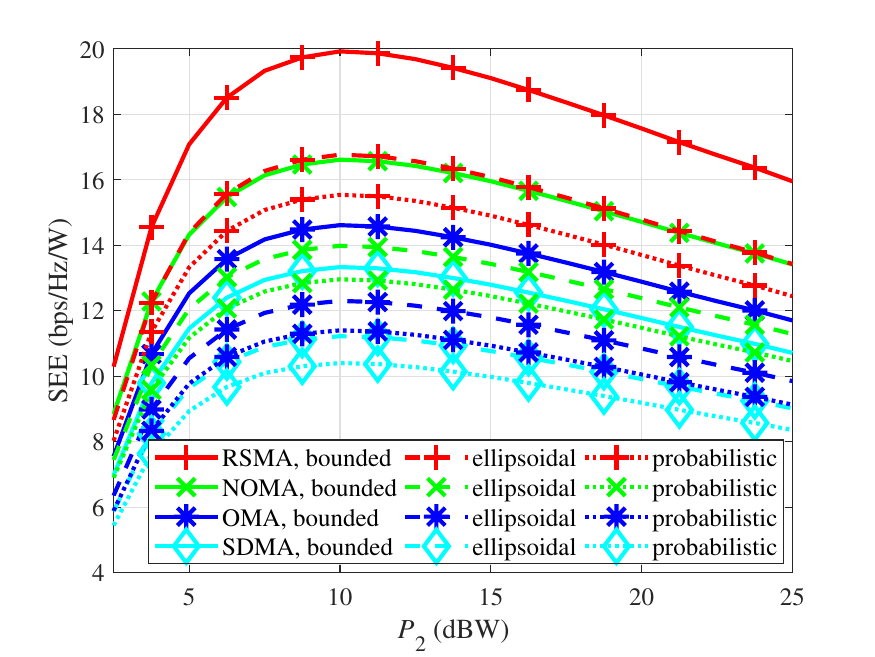}
	\end{minipage}%
}%
  \subfigure[2D beampattern of common stream.]{
  	\label{common_stream_2D}
  	\begin{minipage}[t]{0.25\linewidth}
  		\centering
  		\includegraphics[width=2in]{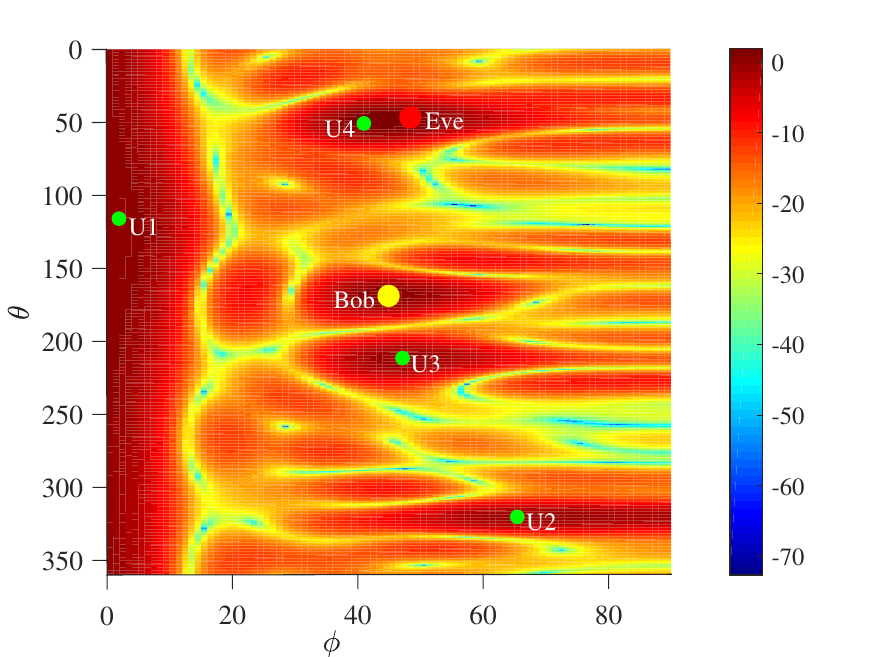}
  	\end{minipage}%
  }%
    \subfigure[2D beampattern of private stream.]{
    	\label{private_stream_2D}
    	\begin{minipage}[t]{0.25\linewidth}
    		\centering
    		\includegraphics[width=2in]{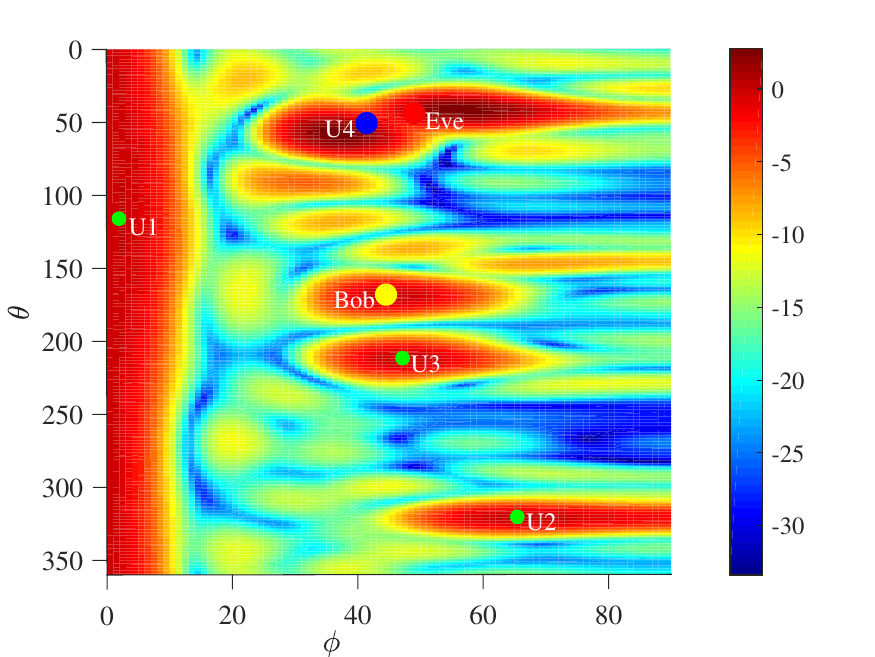}
    	\end{minipage}
    }%
	\centering
	\caption{Robustness validation for adversarial models and RSMA beampattern illustration.}
	\label{four_fig}
\end{figure*}

The convergence performance of the SEE improvement in the heterogeneous ISAC is shown by the solid line in Fig. \ref{iteration_com}, and the traversal algorithm is depicted by the dashed line in contrast.
For a fixed  $M_1$, the SEE gradually increases with the increase of the number of iterations. After an acceptable number of iterations, SEE converges. For a given number of iterations, SEE increases gradually with the increase of $M_1$, since the transmit power carried by a single antenna is constant. The greater $M_1$ is, the greater the maximum transmit power becomes, which is translated to an increase of the spatial degrees of freedom and BF design accuracy and efficiency, leading to a greater $R_{\rm S}$ for a fixed transmit power. Then, SEE also increases. As $M_1$ increases, the number of iterations required to solve the SEE gradually increases. This is because as $M_1$ increases, the computational complexity of the solution also increases; thus the number of iterations increases. For Algorithm \ref{algorithm_4} and traversal algorithm, as $M_1$ increases, the numbers of iterations required to achieve convergence are 24, 30, 39, and 30, 40, and 52, respectively. Obviously, compared with the traversal algorithm, under similar reliability condition and for fixed $M_1$, Algorithm \ref{algorithm_4} requires fewer iterations to achieve convergence, achieving better convergence performance and effectiveness. This experimental result corroborates our theoretical analysis and validates the practical convergence of the proposed Algorithm \ref{algorithm_4}.


In Fig. \ref{SCNR_com}, we show the relationship between $\gamma_{\rm b1}$ and $P_{\rm 1}$ under different BF optimization schemes. It can be seen that under the same BF scheme, $\gamma_{\rm b1}$ increases with the increase of $P_{\rm 1}$, since the signal power received by the target and reflected to $\rm BS_1$ increases with the increase of $P_{\rm 1}$ for a given RCS. Meanwhile, the clutter power and noise power in the environment are relatively stable, and thus $\gamma_{\rm b1}$ enhances. In addition, the proposed RSMA-based BF scheme outperforms NOMA in \cite{ISAC_PLS_46_r}, TDMA in \cite{ISAC_PLS_54}, and non-optimized echo signal in \cite{ISAC_PLS_13}. This is because the proposed SEE BF optimization scheme continuously receives the echo signal in the time dimension, optimizes the BF vector of the echo signal in the space dimension, and simultaneously uses the sensing signal and the communication signal for sensing the target in the power dimension. However, the other three schemes either neglect to optimize the BF vector of the echo signal, or do not continuously or stably receive the target information, or only use the sensing signal for target identification. Therefore, the proposed multi-BF optimization scheme outperforms the contributions of \cite{ISAC_PLS_46_r,ISAC_PLS_54,ISAC_PLS_34} in terms of sensing.

On the one hand, Fig. \ref{SEE_ICSI_com} presents the relationship between SEE and $P_{\rm 2}$ under different multiple access schemes. Under the same multiple access scheme, with the increase of $P_{\rm 2}$, SEE shows a change process of first increasing and then decreasing. This phenomenon is reasonable and based on the definition of SEE \cite{ISAC_PLS_13}. The second derivative of SEE with respect to $P_{\rm 2}$ is less than 0, implying SEE is a convex function of $P_{\rm 2}$. For a given $P_{\rm 2}$, RSMA obtains the best SEE among four multiple access schemes.
On the other hand, Fig. \ref{SEE_ICSI_com} demonstrates the relationship between SEE and $P_{\rm 2}$ under different CSI uncertainty parameters. As can be observed, when the legitimate link CSI errors are introduced ($ e_{\rm h}$=$ e_{\rm g}$=0.0, 0.1, 0.2), the SEE experiences a modest degradation of 8-12\% compared to the baseline case with $ e_{\rm h}$=$ e_{\rm g}$=0.0.  This demonstrates that the proposed scheme maintains reasonable performance even under more realistic CSI conditions. Crucially, the performance advantage of our proposed RSMA scheme over benchmark schemes remains consistent across all CSI uncertainty levels. RSMA-based approach consistently shows 18-25\% higher SEE than NOMA and 30-35\% higher SEE than SDMA and OMA, regardless of the comprehensive CSI error conditions. The constructed green interference maintains its effectiveness in degrading Eve's reception while protecting Bob's confidential communication, even when all channel links have estimation errors. This confirms the inherent robustness of the security mechanism against practical implementation challenges.
Moreover, given that SEE is related to the difference between Bob’s and Eve’s rates, uncertainty in Eve’s CSI directly affects the numerator of SEE expression, i.e., the potential eavesdropping capacity.  An underestimation of Eve’s channel gain could severely compromise security.  Therefore, our robust design prioritizes safeguarding against uncertainties in the Eve's links.

As shown in Fig. \ref{SCNR_RS}, the Pareto front visualization reveals the fundamental performance boundary between sensing SCNR and security rate in our heterogeneous ISAC system. The convex shape of the frontier demonstrates the inherent trade-off between these two critical objectives, with each point representing an optimal solution where neither metric can be improved without degrading the other. The color gradient (blue to red) indicates the weighted parameter λ₂ (0.1~0.3, 0.4~0.6, 0.7~0.9), showing how different weight allocations shift the operating point along the efficiency frontier. The calculated trade-off slope of -2.8 dB/(bps/Hz) in the balanced region provides quantitative evidence of the performance exchange rate, indicating that approximately 2.8 dB of SCNR must be sacrificed for each 1 bps/Hz improvement in security rate under typical operating conditions. The well-defined Pareto boundary confirms that our optimization framework effectively captures the fundamental performance limits, while the continuous nature of the frontier demonstrates the scheme’s ability to achieve smooth transitions between sensing-dominant and security-dominant operational modes.

Fig. \ref{SEE_opti_com} presents the relationship between SEE and $P_{\rm 1}$ under different transmit BF optimization schemes. For a fixed multiple access scheme, as $P_{\rm 1}$ increases, the SEE shows a change process of first increasing and then decreasing. The effect of $P_{\rm 1}$ on SEE can be referred to that of $P_{\rm 2}$ on SEE. For a given $P_{\rm 1}$, compared with the optimization of $\textbf{w}_{\rm bo}$ in \cite{ISAC_PLS_48}, optimization of $\textbf{w}_{\rm ta}$ in \cite{ISAC_PLS_33}, and no optimization in \cite{ISAC_PLS_54}, the proposed joint optimization of $\textbf{w}_{\rm bo}$ and $\textbf{w}_{\rm ta}$ achieves a higher SEE. 
Optimization of $\textbf{w}_{\rm ta}$ transforms the harmful interference that affects communication into the green interference that affects eavesdropping with a smaller eavesdropping rate. Optimization of $\textbf{w}_{\rm bo}$ improves the signal quality received by Bob and $R_{\rm S}$. Finally, the proposed joint optimization of $\textbf{w}_{\rm bo}$ and $\textbf{w}_{\rm ta}$ fully demonstrates the advantage in enhancing SEE compared to schemes in \cite{ISAC_PLS_48,ISAC_PLS_33,ISAC_PLS_54}.

As shown in Fig. \ref{SEE_antenna}, our expanded antenna configuration simulations demonstrate several key patterns. Firstly, we observe diminishing marginal returns, i.e., increasing ($M_1$,$M_2$) from (4,4) to (12,12) improves SEE by approximately 198.87\%, while further increasing to (16,12) yields only 9.99\% additional gain. When the antenna configuration is set as (16,16), the SEE stops to increase but instead experiences a 4.52\% decline compared to (16,12). Secondly, as mentioned in the first point, the reasonable performance occurs in balanced configurations (12,12). Compared with the balanced configurations (12,12), the configuration (16,12), that is, the increase in $M_1$, still achieves performance gain. However, when configured as (16,16), the increase in $M_2$ actually leads to a decrease in performance gain. The above situation indicates that $M_1$ plays a more important role in improving SEE. Thirdly, we consider a more practical model, where the circuit power consumption model is positively correlated with ($M_1$,$M_2$), introducing a crucial trade-off that explains the saturation behavior.

Fig. \ref{SEE_user} depicts the relationship between the number of users $N$ and the SEE. It can be observed that, for a fixed $P_2$, as $N$ gradually increases from 2, SEE also gradually increases. When $N$ reaches 12, SEE reaches its peak. Subsequently, as $N$ continues to increase, SEE gradually decreases. Regarding the observed increase of SEE with $N$ increasing, this is a distinctive feature of our proposed green interference scheme. Unlike conventional systems where additional users typically degrade performance due to increased resource competition, our RSMA-based framework strategically leverages the multi-user signals to construct beneficial green interference against the eavesdropper. However, the situation where the security enhancement brought about by the increase in $N$ exceeds the moderate power increase required to serve additional users is not unconditional. As $N$ exceeds a certain threshold, SEE will gradually decrease when the enhanced security cannot fully compensate for the increased power required to serve additional users.

Considering that in the absence of precise business requirement information, adopting a balanced threshold configuration is a robust design strategy, Fig. \ref{SEE_thresholds} explores the system performance when the three thresholds are set to the same value and clearly indicates the optimal operating point. The peak performance was achieved at $I_{\rm S}$=$I_{\rm c}$= $I_{\rm p}$ = 1.0 bps/Hz, indicating that balanced quality of service requirements can produce the best system synergy. Balanced configuration prevents any service dimension from becoming a performance bottleneck, enabling the optimization algorithm to evenly allocate resources across all dimensions. Moreover, the performance degradation within the range of 1.0-1.7 bps/Hz is relatively gentle, providing sufficient flexibility for the actual system configuration.

As shown in Fig. \ref{SEE_latency}, latency analysis confirms that the proposed scheme meets critical 5G requirements. With 0.84 ms processing time for $M_1$=12, the scheme comfortably satisfies both URLLC ($\le$1 ms) and eMBB ($\le$4 ms) latency budgets. The latency scaling with antenna count remains manageable, reaching 1.52 ms for $M_1$=20 configurations. The additional processing for green interference construction adds only 0.18 ms compared to conventional beamforming, representing a reasonable trade-off for the achieved security and energy efficiency improvements. The computational overhead of green interference is well within 5G latency constraints, ensuring practical deployability in real-time systems.

Fig. \ref{SEE_memory} demonstrates that the proposed scheme is implementable on the existing BS hardware. Memory requirement is 94 MB that is well within typical 4-8 GB BS memory. Processing throughput is 2.74 GFLOPs, that is achievable with modern DSPs. Standard is fully compatible with 3GPP Release 16/17 frameworks. As shown in Fig. \ref{SEE_memory}, hardware requirement analysis confirms practical implementability on contemporary BS platforms. The proposed scheme requires 94 MB memory and 2.74 GFLOPs processing capability, well within specifications of commercial BS processors (e.g., Xilinx Zynq UltraScale+ RFSoC). The memory requirement represents only 2.4\% of typical 4 GB BS memory, while the processing requirement is 58\% of typical 5 GFLOPs DSP capacity. The computational demands of green interference are compatible with existing 5G infrastructure, requiring no hardware upgrades.

As illustrated in Fig. \ref{SEE_cross_interference}, multi-cell deployment analysis reveals that green interference actually improves cross-cell interference characteristics. The proposed scheme achieves average SINR of 18.2 dB and cell-edge SINR of 8.7 dB, outperforming both artificial noise (16.4 dB average, 7.1 dB cell-edge), RIS-security (17.3 dB average, 8.2 dB cell-edge), and conventional schemes (15.8 dB average, 6.9 dB cell-edge). This counterintuitive improvement stems from the directed nature of green interference, which minimizes spillover to adjacent cells compared to omnidirectional artificial noise. Properly designed green interference reduces rather than increases cross-cell interference, addressing a fundamental concern in multi-cell deployments.

As demonstrated in Fig. \ref{SEE_runtime_antenna}, the runtime characteristics exhibit distinct complexity scaling patterns across different schemes. The proposed green-RSMA scheme requires approximately 3.2 seconds for ($M_1$,$M_2$)=(12,12) antennas, which represents a practical compromise between performance and computational efficiency. In comparison, RIS-security exhibits the steepest complexity growth, requiring 4.6 seconds for (12,12) antennas and escalating rapidly to 28.3 seconds for (32,32).  NOMA and RSMA in \cite{ISAC_PLS_48} show moderate complexity, with runtimes of 2.8 and 3.1 seconds for (12,12). OMA and SDMA demonstrate the most favorable runtime characteristics, completing in 2.1 and 1.9 seconds respectively for (12,12), albeit at the cost of significantly reduced SEE performance.

Furthermore, to fully demonstrate the topological robustness of our proposed scheme, we simulate and verify the differences in SEE between the RSMA adopted and the NOMA, OMA, and SDMA as comparison objects in three typical scenarios, i.e., centralized, edge-clustered, and hybrid deployments, shown in Figs. \ref{central_de}, \ref{edge_de}, \ref{hybrid_dp}, and the comparison regarding SEE is displayed in Fig. \ref{SEE_topology}. Results demonstrate that while absolute performance metrics vary across these topologies with the edge-clustered scenario showing approximately 18\% lower SEE compared to the centralized case due to more severe interference conditions. The proposed RSMA-based scheme consistently maintains its performance advantage over benchmark schemes (NOMA, OMA, SDMA) in all tested scenarios. This robustness stems from RSMA's inherent capability to manage inter-user interference through intelligent rate-splitting and our optimized BF design that dynamically adapts to user distribution patterns.

As illustrated in Fig. \ref{channel_model_com}, the proposed RSMA-based green interference scheme consistently achieves the highest SEE across all three channel fading models, which validates its robustness against different fading characteristics. However, the absolute performance and the performance gap between different schemes are highly dependent on the channel model. For the performance in the Rician-shadowed fading, in this model, which features a dominant (though potentially shadowed) LoS path, all schemes perform at their best. The proposed scheme leverages the stable spatial characteristics afforded by the LoS component to precisely shape beams and construct effective green interference, resulting in the highest SEE. For the performance in Nakagami-2 fading, when the channel lacks a deterministic LoS component, we observe a noticeable performance degradation for all schemes. For the performance in severe fading (Nakagami-1/Rayleigh), under the most severe Rayleigh fading condition, the SEE for all schemes drops substantially.

To further verify the robustness of the adopted RSMA, we compared the SEE under bounded, ellipsoidal, and probabilistic models in Fig. \ref{uncertainty_model}, where simulations validate that the performance degradation remains bounded with SEE reduction limited to 22\% in extreme conditions. Furthermore, it can be concluded from Fig. \ref{uncertainty_model} the proposed green-RSMA scheme has achieved the best SEE among the three uncertainty models.


In Figs. \ref{common_stream_2D} and \ref{private_stream_2D}, yellow, red, green, and blue dots are positions for Bob, Eve, users, and the user wanting this private stream.
Fig. \ref{common_stream_2D} depicts the \emph{two-dimensional} (2D) beampatterns of the common stream. There are 4 significant and relatively close amplitude gains with values greater than $\rm 0\,dB$. This is because $\rm BS_2$ serves 4 communication users, all of whom require the common stream. Therefore, the amplitude gain of the common stream for 4 users is greater than $\rm 0\,dB$, ensuring that the power of the transmit common stream gets amplified. On the other hand,, amplitude gains of common stream transmitted to Bob and Eve is close to $\rm 0\,dB$. The optimization of $\textbf{o}_{\rm c}$ in this paper improves the signal power received by each user, and overcomes the fading problem in the signal transmission process. Therefore, each user obtains the common stream reliably while the interference from the common stream exerted on the signal reception of Bob and Eve is limited.

Fig. \ref{private_stream_2D} illustrate the 2D beampatterns of the private stream. In Fig. \ref{private_stream_2D}, it can be found that among 6 amplitude gains, the amplitude gain of the fourth user is significantly greater than $\rm 0\,dB$, and the amplitude gains of other uses are close to or less than $\rm 0\,dB$. This is because a private stream is only useful for one user and harmful to other users. Therefore, when a private stream is provided to the fourth user, it is necessary to increase the amplitude gain of the fourth user and enhance the signal power received by the fourth user. Meanwhile, the amplitude gains of the other users is cut down, decreasing the received interference power, and thus improving the received SINR. In addition, the amplitude gain for Bob is less than $\rm 0\,dB$, and the interference received by Bob is suppressed, while the amplitude gain for Eve is greater than $\rm 0\,dB$, causing a stronger interference to Eve to improve the SEE. Hence, the optimization of $\textbf{o}_n$ improves the SINR of users and the SEE of the system.



\section{Conclusion}
This paper proposed a heterogeneous ISAC employing ISAC communication, sensing as well as  multicast communications. Given the challenges posed by interference, eavesdropping and imperfect CSI, we qualified sensing, communication, security, and EE. To enhance these metrics, we developed a green interference scheme by leveraging the ISAC sensing signal and the RSMA communication signal to maximize the SEE and study the trade-off between sensing and security, which was realized in the form of multi-BF optimization at the mathematical level. 
Since the original optimization problem was highly non-convex, posing significant challenges for reliable and efficient solutions, we decomposed it into three tractable sub-optimization problems using Taylor series expansion, MM, SDP, and SCA, and then solved it with alternating optimization. Finally, simulations validated the superior efficiency, robustness, and scalability
of our designs in achieving secure and green sensing and communications in challenging scenarios in presence of interference, eavesdropping and imperfect CSI.



\begin{appendices}

\section{Derivation of Bounds for CSI Error Vectors}
The signal power term in the SINR expressions like $\left\| {{\bf{h}}_{\rm{e}}^{\rm{H}}{{\bf{w}}_{{\rm{bo}}}}} \right\|_2^2$ involve a quadratic form with uncertain channel
\begin{equation}
	\left\| {{{\left( {{{\bf{h}}_{{\rm{es}}}} + {{\bf{h}}_{{\rm{er}}}}} \right)}^{\rm{H}}}{{\bf{w}}_{{\rm{bo}}}}} \right\|_2^2 = {\bf{w}}_{{\rm{bo}}}^{\rm{H}}\left( {{{\bf{h}}_{{\rm{es}}}} + {{\bf{h}}_{{\rm{er}}}}} \right){\left( {{{\bf{h}}_{{\rm{es}}}} + {{\bf{h}}_{{\rm{er}}}}} \right)^{\rm{H}}}{{\bf{w}}_{{\rm{bo}}}}.
\end{equation}
To find a deterministic bound for this quadratic term, we manipulate its expansion as $\left\| {{\bf{h}}_{\rm{e}}^{\rm{H}}{{\bf{w}}_{{\rm{bo}}}}} \right\|_2^2 = \left\| {{\bf{h}}_{{\rm{es}}}^{\rm{H}}{{\bf{w}}_{{\rm{bo}}}} + {\bf{h}}_{{\rm{er}}}^{\rm{H}}{{\bf{w}}_{{\rm{bo}}}}} \right\|_2^2$.
	By applying the triangle inequality ${\left\| {{\bf{a}} + {\bf{b}}} \right\|_2} \le {\left\| {\bf{a}} \right\|_2} + {\left\| {\bf{b}} \right\|_2}$ and the property $\left\| {{\bf{a}} + {\bf{b}}} \right\|_2^2 \le \left\| {\bf{a}} \right\|_2^2 + \left\| {\bf{b}} \right\|_2^2 + 2{\left\| {\bf{a}} \right\|_2}{\left\| {\bf{b}} \right\|_2}$, we can derive the corresponding upper and lower bounds. A standard and equivalent approach is to bound the perturbed quadratic form ${\bf{w}}_{{\rm{bo}}}^{\rm{H}}\left( {{{\bf{h}}_{{\rm{es}}}}{\bf{h}}_{{\rm{es}}}^{\rm{H}} + {{\bf{\Delta }}_{\rm{h}}}} \right){{\bf{w}}_{{\rm{bo}}}}$, where ${{\bf{\Delta }}_{\rm{h}}} = {{\bf{h}}_{{\rm{es}}}}{\bf{h}}_{{\rm{er}}}^{\rm{H}} + {{\bf{h}}_{{\rm{er}}}}{\bf{h}}_{{\rm{es}}}^{\rm{H}} + {{\bf{h}}_{{\rm{er}}}}{\bf{h}}_{{\rm{er}}}^{\rm{H}}$.
	The spectral norm of the perturbation matrix ${{\bf{\Delta }}_{\rm{h}}}$ is bounded using the triangle inequality for matrices and the sub-multiplicative property ${\left\| {{{\bf{\Delta }}_{\rm{h}}}} \right\|_2} \le 2{{{e}}_{\rm{h}}}{\left\| {{{\bf{h}}_{{\rm{es}}}}} \right\|_2} + {{e}}_{\rm{h}}^2 = {{{e}}_{{\rm{h,UB}}}}$. This bound ${{{e}}_{{\rm{h,UB}}}}$ is the one used in the manuscript to define the extreme matrices. 
	Similarly, the above derivations regarding ${{{\bf{\Delta }}_{\rm{h}}}}$ and ${{{e}}_{{\rm{h,UB}}}}$ are also applicable to ${{{\bf{\Delta }}_{\rm{g}}}}$ and ${{{e}}_{{\rm{g,UB}}}}$.

\section{Proof of Proposition 1}
Based on (\ref{defi_RS}), (\ref{P1}\rm e) can be rewritten as
\begin{equation}\label{rate_to_SINR}
\begin{array}{l} \displaystyle
	\frac{{{\rm{Tr}}\left( {{{\bf{H}}_{{\rm{bo,min}}}}{{\bf{W}}_{{\rm{bo}}}}} \right)}}{{{\rm{Tr}}\left( {{{\bf{G}}_{{\rm{bo,max}}}}{{\bf{O}}_{{\rm{c,}}n}}} \right) + {\sigma ^2}}} - \left( {\tau  - 1} \right) \ge \\ \displaystyle
	\tau \frac{{{\rm{Tr}}\left( {{{\bf{H}}_{\rm{e}}}{{\bf{W}}_{{\rm{bo}}}}} \right)}}{{{\rm{Tr}}\left( {{{\bf{H}}_{\rm{e}}}{{\bf{W}}_{{\rm{ta}}}}} \right) + {\rm{Tr}}\left( {{{\bf{G}}_{\rm{e}}}{{\bf{O}}_{{\rm{c,}}n}}} \right) + {\sigma ^2}}}.
\end{array}
\end{equation}
Then, we transform (\ref{rate_to_SINR}) into a LMI as
\begin{equation}\label{non-LMI}
\begin{array}{*{20}{l}} \displaystyle
	{{\rm{Tr}}\left( {{{\bf{H}}_{{\rm{bo,min}}}}{{\bf{W}}_{{\rm{bo}}}}} \right)\left[ {{\rm{Tr}}\left( {{{\bf{H}}_{\rm{e}}}{{\bf{W}}_{{\rm{ta}}}}} \right) + {\sigma ^2} + {\rm{Tr}}\left( {{{\bf{G}}_{\rm{e}}}{{\bf{O}}_{{\rm{c,}}n}}} \right)} \right]}\\  \displaystyle
	{ - \left( {\tau  - 1} \right)\left[ {{\rm{Tr}}\left( {{{\bf{G}}_{{\rm{bo,max}}}}{{\bf{O}}_{{\rm{c,}}n}}} \right) + {\sigma ^2}} \right]}\\  \displaystyle
	{ \times \left[ {{\rm{Tr}}\left( {{{\bf{H}}_{\rm{e}}}{{\bf{W}}_{{\rm{ta}}}}} \right) + {\rm{Tr}}\left( {{{\bf{G}}_{\rm{e}}}{{\bf{O}}_{{\rm{c,}}n}}} \right) + {\sigma ^2}} \right]}\\  \displaystyle
	{ \ge \tau {\rm{Tr}}\left( {{{\bf{H}}_{\rm{e}}}{{\bf{W}}_{{\rm{bo}}}}} \right)\left[ {{\rm{Tr}}\left( {{{\bf{G}}_{{\rm{bo,max}}}}{{\bf{O}}_{{\rm{c,}}n}}} \right) + {\sigma ^2}} \right].}
\end{array}
\end{equation}
After some algebraic manipulations, we obtain and prove
\begin{equation}\label{LMI-1}
{U_1}{U_2} \ge \frac{\tau }{{\tau  - 1}}{\rm{Tr}}\left( {{{\bf{H}}_{{\rm{bo,min}}}}{{\bf{W}}_{{\rm{bo}}}}} \right){\rm{Tr}}\left( {{{\bf{H}}_{\rm{e}}}{{\bf{W}}_{{\rm{bo}}}}} \right).
\end{equation}

\section{Proof of Proposition 2}\label{pro3}
Based on Proposition 1, we hold that
\begin{equation}\label{pro3-prove-1}
	\begin{array}{l} \displaystyle
		{\left( {{U_1} + {U_2}} \right)^2} - {\left( { - {U_1} + {U_2}} \right)^2} = 4{U_1}{U_2}\\ \displaystyle
		\ge \frac{{4\tau }}{{\tau  - 1}}{\rm{Tr}}\left[ {{{\bf{H}}_{{\rm{bo,min}}}}{{\bf{W}}_{{\rm{bo}}}}} \right]{\rm{Tr}}\left[ {{{\bf{H}}_{\rm{e}}}{{\bf{W}}_{{\rm{bo}}}}} \right]\\ \displaystyle
		= {\left[ {2\sqrt {{\tau  \mathord{\left/
							{\vphantom {\tau  {\left( {\tau  - 1} \right)}}} \right.
							\kern-\nulldelimiterspace} {\left( {\tau  - 1} \right)}}} {\rm{Tr}}\left( {{{\bf{h}}_{{\rm{bo}}}}{\bf{h}}_{\rm{e}}^{\rm{H}}{{\bf{W}}_{{\rm{bo}}}}} \right)} \right]^2},
	\end{array}
\end{equation}
or equivalently
\begin{equation}\label{pro3-prove-2}
	\begin{array}{l}
		{U_1} + {U_2} \ge \\   \displaystyle
		\sqrt {{{\left[ {2\sqrt {{\tau  \mathord{\left/
									{\vphantom {\tau  {\left( {\tau  - 1} \right)}}} \right.
									\kern-\nulldelimiterspace} {\left( {\tau  - 1} \right)}}} {\rm{Tr}}\left( {{{\bf{h}}_{{\rm{bo}}}}{\bf{h}}_{\rm{e}}^{\rm{H}}{{\bf{W}}_{{\rm{bo}}}}} \right)} \right]}^2} + {{\left( { - {U_1} + {U_2}} \right)}^2}}.
	\end{array}
\end{equation}
According to the definition of \emph{second order cone} (SOC) and (\ref{pro3-prove-2}), we obtain (\ref{pro3-1}). This concludes proof of Proposition \ref{pro_re2}.

\section{Proof of Proposition \ref{pro_re_3}}
According to (\ref{PB_C1}) and (\ref{pro3-prove-2}), we get
\begin{equation}\label{pro4-prove-1}
	\begin{array}{l} \displaystyle
		{\rm{Tr}}\left( {{{\bf{h}}_{{\rm{bo}}}}{\bf{h}}_{\rm{e}}^{\rm{H}}{{\bf{W}}_{{\rm{bo}}}}} \right) \le {\rm{Tr}}\left[ {\left( {{{\bf{h}}_{{\rm{bo}}}}{\bf{h}}_{{\rm{es}}}^{\rm{H}} + {{\bf{h}}_{{\rm{bo}}}}{\bf{h}}_{{\rm{er}}}^{\rm{H}}} \right){{\bf{W}}_{{\rm{bo}}}}} \right]\\ \displaystyle
		\le {\rm{Tr}}\left[ {\left( {{{\bf{h}}_{{\rm{bo}}}}{\bf{h}}_{{\rm{es}}}^{\rm{H}} + {{\rm{e}}_{{\rm{h,UB}}}}{{\bf{h}}_{{\rm{bo}}}}{{\bf{I}}_{1 \times {M_1}}}} \right){{\bf{W}}_{{\rm{bo}}}}} \right]\\ \displaystyle
		\le {\rm{Tr}}\left[ {\left( {{{\bf{h}}_{{\rm{bo}}}}{\bf{h}}_{{\rm{es}}}^{\rm{H}} + {{\rm{e}}_{{\rm{h,UB}}}}{{\bf{I}}_{{M_1} \times {M_1}}}} \right){{\bf{W}}_{{\rm{bo}}}}} \right] = {\rm{Tr}}\left( {{\bf{\tilde H}}{{\bf{W}}_{{\rm{bo}}}}} \right).
	\end{array}
\end{equation}
To ensure that the constraint on the security rate in (\ref{pro3-1}) is always satisfied, based on Proposition 2, (\ref{range_UB}), (\ref{range_LB}), and (\ref{pro4-prove-1}), (\ref{pro3-1}) can be rewritten as
\begin{equation}\label{pro4-prove-2}
	\begin{array}{l}
		{U_1} + {U_{2,\rm LB}} \ge \\   \displaystyle
		\sqrt {{{\left[ {2\sqrt {{\tau  \mathord{\left/
									{\vphantom {\tau  {\left( {\tau  - 1} \right)}}} \right.
									\kern-\nulldelimiterspace} {\left( {\tau  - 1} \right)}}} {\rm{Tr}}\left( {{\bf{\tilde H}}{{\bf{W}}_{{\rm{bo}}}}} \right)} \right]}^2} + {{\left( { - {U_1} + {U_{2,\rm UB}}} \right)}^2}}.
	\end{array}
\end{equation}
Afterwards, based on the definition of the SOC, we can get (\ref{pro4-1}). This concludes the proof.

\end{appendices}

\bibliographystyle{IEEEtran}
\bibliography{ref_PR-ISAC_CRN}

\end{document}